\documentclass{article}
\usepackage{graphicx} 
\usepackage{amsthm}
\usepackage{amsmath}
\usepackage{hyperref}
\usepackage[capitalise]{cleveref}
\usepackage{csquotes}
\usepackage{bbm}

\usepackage[letterpaper,top=2cm,bottom=2cm,left=3cm,right=3cm,marginparwidth=1.75cm]{geometry}

\newtheorem{itheorem}{Theorem}

\newtheorem{btheorem}{Theorem}

\newtheorem{theorem}{Theorem}[section]

\newtheorem{lemma}[theorem]{Lemma}
\newtheorem{proposition}[theorem]{Proposition}
\newtheorem{claim}[theorem]{Claim}
\newtheorem{definition}[theorem]{Definition}

\usepackage{amssymb}

\usepackage{amsmath}
\newcommand{\trans}{{\mathcal{C}}}
\newcommand{\EIP}[0]{EIP-1559 } 

\newcommand{\DOUBLEBLIND}[1]{#1} 

\usepackage[dvipsnames]{xcolor}
\newcommand{\mbe}[1]{{\color{red} #1}} 

\newcommand{\comment}[1]{}

\title{On the Welfare of EIP-1559 with Patient Bidders}

\DOUBLEBLIND{\author{Moshe Babaioff\footnote{Hebrew University of Jerusalem.} \and Noam Nisan\footnote{Hebrew University of Jerusalem and Starkware. } }
}

\begin{document}

\maketitle

\begin{abstract}The ``EIP-1559 algorithm'' is used by the Ethereum blockchain to assemble transactions into blocks. While prior work has  studied it under the assumption that bidders are ``impatient'', we analyze it under the assumption that bidders are ``patient'', which better corresponds to the fact that unscheduled transactions remain in the mempool and can be scheduled at a later time. We show that with ``patient'' bidders, this algorithm produces schedules of near-optimal welfare, provided  it is given a mild resource augmentation (that does not increase with the time horizon).  We prove some generalizations of the basic theorem, establish lower bounds that rule out several candidate improvements and extensions, and propose several questions for future work.
\end{abstract}

\section{Introduction}

\subsection{Background on Blockchains}

Blockchains like Bitcoin \cite{bitcoin} or Ethereum \cite{eth} operate by repeatedly aggregating transactions into blocks.  As a first approximation,
each transaction requires some ``size'' and has some value for its user, and 
each block has some size limit, where the notion of ``size'' is defined by the
blockchain.\footnote{E.g. ``gas'' for Ethereum or (virtual-)bytes for Bitcoin.}
Each block is assembled by a different ``operator''\footnote{In Bitcoin these are the ``miners''  and in Ethereum these are the ``validators''.} 
which the blockchain protocol designer
tries to motivate to maximize the
total value 
subject to the block's capacity constraints.  A key difficulty here is that both users and operators are selfish and strategic,  
so the mechanism that specifies which transactions are accepted to the block and how much do their users pay and to whom
must take such strategic behavior into account.  A large literature has studied such mechanisms.  While Bitcoin's ``pay your bid'' mechanism is relatively simple (but
not incentive compatible for the users), Ethereum's mechanism \cite{1559}, known as EIP-1559, is more sophisticated and, following \cite{R21}, 
a large literature (e.g. \cite{LMRSP21,CS23,GTW24,AGHH23}) studies its
strategic properties.

In this paper we undertake an algorithmic analysis of this protocol and the family of 
protocols that vary its various parameters.   
Specifically, we ask to what extent it succeeds in maximizing the social welfare given the size constraints.  Our main result is essentially positive: for ``patient bidders'', the social welfare achieved by EIP-1559, 
with a slight tweak of parameters, is close to optimal.  Before specifying more precisely what ``close to optimal'' means, we need to be more specific about our time-model of transactions.  

Blocks are created one at a time\footnote{For Ethereum a block every 12 
seconds and for Bitcoin a block every 10 minutes on the average.} and
transactions arrive continuously through time.  Much analysis of blockchains assumes that a transaction that arrives at a certain point in time {\em must} 
be scheduled in the immediately next block or loses all value for its user.  This is called the model of ``impatient users''.  In reality, however, blockchains
operate under an opposite assumption and submitted transactions remain in the ``mempool'' until accepted, 
and thus may be scheduled at any future block.  
The model that assumes that transactions can be scheduled at any future time
without losing their value, which is closer to how blockchains actually operate, is termed a 
model with 
``patient users'',
and has only received modest attention \cite{N23, HLM21, PS24, GY24}. 
A more general model may consider a mix of patient, impatient, and also ``partially patient'' users, who may have some time-sensitivity, e.g., model the value of a transaction as decreasing with its execution time \cite{HLM21, PS24, GY24}. 
Unfortunately, current mainstream 
blockchain transaction-fee mechanisms
do not allow the user to express any time preferences.
We show
the near-optimality of EIP-1559 for fully patient users in the worst case, i.e., without any stochastic
or distributional assumptions on the demand.\footnote{
It is not difficult to observe that if one looks, instead, at {\em impatient users}, then 
any algorithm that {\em relies on historical prices} to set 
inclusion criteria for the current block (as does EIP-1559) may catastrophically 
fail to achieve any decent social welfare under adversarial 
conditions.  
We will also show 
that no {\em online algorithm} can achieve near-optimality for a mix of patient and either impatient or {\em partially patient} users (even 
if they are ``very patient''); see \Cref{sec:intro-imp}.}

\subsection{The EIP-1559 algorithm and its parameters}

The EIP-1559 algorithm assembles a sequence of blocks {with {\em target average size} of $B$ per block.}%
\footnote{Ethereum's choice is $B=18M$ 
gas.}  
I.e. the original constraint of maximum block size is relaxed, 
and is replaced by effectively limiting the {\em average} block size to be $B$, while relaxing the  constraint of {\em maximum}
block size to being only slightly larger than $B$.\footnote{Maximum size of only $2B$, a factor $2$ increase that is independent of the parameters of the input (such as the range of per-unit values and the time horizon).}  This relaxation is justified by the specific constraints 
of the Ethereum blockchain\footnote{I.e. it was deemed  that the ``syncing time'' of ``full nodes'' (i.e., reading the entire history of the block chain) is more significant than the maximum block size limitation, and for it, only the average block size is important.} 
and may be viewed as a type of resource augmentation 
which was not
technically studied in the literature of bin-packing and knapsack algorithms, and
may be of independent interest.\footnote{This relaxation may be related in
motivation to open-ended bin packing (e.g. \cite{E22}), but not technically so.}

Here is how it operates: for every block $t$ it chooses a price $p_t$ 
(according to an algorithm that we will explain shortly) and accepts only transactions 
whose value per unit
is at least the price $p_t$.  A difficulty arises if there are too many transactions 
with such high value per unit, 
as the EIP-1559 algorithm does not allow any block to be of total 
size that is larger than $c$ times 
the target block size, with the parameter 
choice $c=2$.\footnote{I.e. the maximum allowed block size 
in Ethereum is $36M$ 
gas.}  In this
case only a subset of the transactions that fit into the maximum block size is chosen, 
with the subset 
determined by an additional ``tip'' bid that is strategically chosen by each user.  
In this paper we will not undertake a strategic analysis, and 
we will analyze the protocol under the assumption that the subset of transactions is being chosen by an {\em adversary}\footnote{This naturally only strengthens our main result.}, making only the following mild assumption: 
the adversary must pick a %
maximal (by inclusion) set of transactions that fit into the allowed {\em maximum} block size.%
\footnote{This is a mild assumption as any rational operator that is myopic (cares only about the current block and disregards the future) 
is never better off by excluding transactions that can simply be added to the block.
In particular, the presence of (non-negative) tips only encourages the inclusion of additional transactions in the block.}

The basic idea of choosing the price $p_t$ for block $t$ is rather simple: if the previous block size was above the target block size $B$,
then we need to increase the price, and if it was below the target then we decrease the price.
The increase or decrease is multiplicative according to the following formula:
$p_{t+1}=p_t \cdot (1+ \eta (Q_t-B)/B)$, where $Q_t$ denotes
the total size of block $t$, 
and $\eta$ is a small constant picked in Ethereum to be 
$\eta=1/8$. This formula should be viewed as an approximation to the theoretically
more elegant price-adjustment formula\footnote{When a similar adjustment mechanism was later adopted for ``blob gas'' 
\cite{4844} then a better approximation to the desired exponential form was taken.}
$p_{t+1}=p_t \cdot e^{\eta \cdot (Q_t-B)/B}$ which we will analyze in this paper.
Effectively, there is also a minimum
price\footnote{This happens in the Ethereum implementation
because all arithmetic is done in integers denominated in Wei, i.e. $10^{-18}$ ETH,
and the price cannot become 0.} 
$p_{min}$,  
and if the previous formula drops below it then this minimum price is taken instead. The
initial price $p_1$ was chosen arbitrarily\footnote{The first Ethereum EIP-1559 
block was given a price of 1 GWei, i.e. $10^{-9}$ ETH.}.  In the
rest of this paper whenever we talk about ``the EIP-1559 algorithm'',
we actually consider the algorithm with any setting 
of these five parameters: the target
block size $B$, the maximum block size $c \cdot B$, the 
price adjustment parameter $\eta$, the minimum price $p_{min}$ and
the initial price $p_1$ (with $B>0, c>1, \eta>0, p_{min} > 0$, and $p_1\geq p_{min}$).

This algorithm produces blocks whose size may vary in the range $0$ to $c \cdot B$. The
{\em target} block size $B$ is reached only in an average sense, but in a rather strong
sense of averaging that we wish to formalize as follows.

\begin{definition} 
{For a monotone non-decreasing function $\Delta:\mathbb{N}\rightarrow \mathbb{N}$,
a schedule is said to have \emph{average block size limit $B$ with slackness $\Delta$,
}} if for any number of times steps $k$, the total size of transactions during any $k$ consecutive blocks 
is at most $(k+\Delta(k)) \cdot B$. 
{When $\Delta$ is a constant function (for constant $Z$ it holds that $\Delta(k)=Z$), we abuse notation and simply write $\Delta=Z$.} 
\end{definition}

The basic property of the
EIP-1559 algorithm 
is that it indeed produces blocks whose average size is the target size $B$.\footnote{Our analysis is for the theoretically
clean version with $p_{t+1}=p_t \cdot e^{(\eta (Q_t-B)/B)}$. As shown in \cite{LRMP22}, the variant with 
$p_{t+1}=p_t \cdot (1+ \eta (Q_t-B)/B)$ results in a slightly larger average block size.}  
While this general fact is well known, we formalize it in our notion of ``average block size limit $B$ with slackness $\Delta$'' as follows:

\begin{proposition}
    Any schedule produced by an EIP-1559 algorithm with parameters
    $(B, c, \eta, p_{min}, p_1)$
    has average block size limit $B$ with slackness 
    $\Delta = \frac{1}{\eta}\cdot \ln \frac{H}{L}+ {(c-1)}$, where $H$
    and $L$
    are upper and lower bounds on the per-unit value 
    of any input transaction, 
    and on the constants $p_1$ and $p_{min}$.
\end{proposition}

For completeness we provide a proof, in our model and notation, 
in \cref{prop:EIP-slack}. 

\subsection{Our Main result}

Our main result states that the EIP-1559 algorithm 
is near-optimal for patient bidders, under some mild relaxations.  
Specifically, 
its allocation over a slightly larger time horizon has social welfare that is
close to that of the optimal, clairvoyantly chosen, 
schedule, that has maximum block size $B$. 

Maximizing the welfare is a challenging problem. Even for a single block it is a knapsack problem, and thus NP-hard. 
On top of the computational issue, there are two additional hurdles that the EIP-1559 algorithm has to overcome.
First, the scheduling problem that the EIP-1559 algorithm faces is an online problem, and so the algorithm is an online algorithm.
While the offline (clairvoyant) optimal benchmark can take into account
future transactions that may
``fit with'' some existing transactions, online
algorithms cannot do so.  Nicely, as we will show, the relaxation of the maximum block
size constraint to an average block size constraint suffices to compensate for this online disadvantage,
provided that the maximum block size limit is at least twice that of the benchmark, 
i.e. $c \ge 2$.

Secondly, the EIP-1559 algorithm is not only an online algorithm, but is further constrained as it is a  ``pricing-based'' algorithm. It has to make decisions that are based on prices (so it cannot pick any allocation), and moreover, the price at each block is based only on on-chain information from previous blocks, without any knowledge of the pending transactions in the mempool.

Such algorithms must ``search'' 
for the correct prices, which may be difficult under adversarial conditions.  
As we show, the
EIP-1559 algorithm manages to do so with a loss that can be overcome by extending the time horizon by the 
number of blocks 
that is logarithmic in the per-unit value range. 
This
single logarithmic loss is global over the whole range of blocks, even though prices may keep fluctuating. 

{Let us use $ALG$ to denote the schedule produced by an \EIP algorithm, and for any schedule $S$ let $SW(S,[1,t])$ denote the total social welfare of $S$ up to time $t$.}
\begin{itheorem}\label{thm:basic:intro}
        Fix an \EIP algorithm with parameters $(B, c, \eta, p_{min},p_1)$ 
        for some $c>2$. 
        Consider any (adversarially chosen) sequence $\trans$ of input transactions where each transaction $i\in \trans$ has per-unit value $v_i$ satisfying
        $e^{\eta} \cdot p_{min} \le v_i \le v_{max}$ 
        for some 
        $v_{max}\ge p_1$. Let $OPT_B$ be an arbitrary schedule with maximum block size $B$.
        Then for any time $T$ it holds that 
          $$SW(ALG,[1,T+\Gamma])
    \ge (1-\delta) \cdot 
    SW(OPT_B,[1,T])
    $$ 
for
     $\delta = 1-e^{-\eta}\leq \eta$ and
     {any integer}
     $\Gamma {\geq 
     \frac{1}
     {\eta}}\cdot \ln\left(\frac{v_{max}}{p_{min}}\right) \cdot \max \{1,\frac{1}{c-2}\} + c$.  
\end{itheorem}

We thus see that the EIP-1559 algorithm can indeed compete with the optimal schedule, but after 
suffering two types of losses: first, there is a loss of a small $\delta \le \eta$ fraction
of the social welfare, where as the adjustment rate $\eta$ approaches $0$ this loss vanishes.  Second,
the algorithm is given $\Gamma$ additional rounds, where $\Gamma$
is inversely proportional to $\eta$.
Conceptually, if $T$ was given in advance 
we could balance between these two losses
(for any fixed $v_{max}, p_{min}$ and $c$)
by choosing $\eta = O(1/\sqrt{T})$, which would lose 
$O(1/\sqrt{T})$ fraction of welfare every step as well as require an additional $O(1/\sqrt{T})$ fraction of steps.\footnote{
This is somewhat in the vein of regret bounds that balance
between losing a small factor 
(that is proportional to the learning rate)
every step and 
suffering an additional additive loss that
is inversely proportional to the learning rate.  
Unlike in classic regret bounds the theorem 
does not combine both losses into a single $O(1/\sqrt{T})$ regret measure.
}

We remark that the EIP-1559 algorithm is not aware of the time horizon ($T$), and in practice it essentially can continue running indefinitely. The guarantee of the theorem holds concurrently for every time horizon $T$ (with  $v_{max}$ being the maximal per-unit value up to $T$), even with respect to the time horizon dependent optimal allocation $OPT(T)$.
  
Before looking more carefully at generalizations and limitations of this theorem, 
it is useful to compare
the result to two similar-in-spirit papers that study more difficult models
and consequently obtain much weaker 
results for their models.  The first is the analysis of \cite{FMPM21} of price-based mechanisms
in the model of impatient buyers. In this model they only obtain a constant factor 
approximation and only under
a stochastic model.  The second is the analysis of \cite{ADM24} of a 
generalization of EIP-1559 to
multiple resources.  
They obtained near-optimality (a regret bound somewhat similar to ours) for the 
case of impatient bidders,
but only relative to the weak benchmark of a single fixed price.

We actually prove a significant generalization of this basic theorem simultaneously 
along several dimensions {(see \cref{thm--ub--eip1559-optimal-frac} for the exact formal result)}:
\begin{enumerate}
    \item {\bf Maximum Block Size vs. Maximum Transaction Size:}
    The limitation on the maximum block size $c>2$ is required only when the largest
    transaction size $q_{max}$ can be as large as the full target block size $B$.  If we
    have a better bound on the maximum transaction size, then the limitation on $c$ can be relaxed
    to $c > 1 + q_{max}/B$, and in the bound on $\Gamma$ {the expression $1/(c-2)$ is replaced by $1/(c-1-q_{max}/B)$ (with a slight change to the additive term $c$ as well)}.  
    Notice that when we have any bound on 
    the maximum transaction size $q_{max}/B < 1$ then the actual choice made in Ethereum
    of $c=2$ suffices for obtaining a good approximation and losing only $\eta=1/8$ fraction
    of social welfare when allowed a small number of extra blocks.  

    \item {\bf A Stronger Benchmark with Average Block Size Constraints:}
    As mentioned above, the EIP-1559 algorithm's schedule 
    has {\em average block size} $B$, a constraint that is more relaxed than
    having {\em maximum block size} $B$ which is the benchmark in the theorem. 
    We extend the theorem to show that the algorithm also competes with
    the benchmark of schedules with {\em average block size}
    $B$ with any slackness $\Delta$, where the bound on the constant $\Gamma$ is 
    now increased by $\Delta$.  In this generalization, the EIP-1559 algorithm with 
    maximum block size $c > 1+q_{max}/B$ competes with any schedule with average block size $B$,
    without requiring any upper bound on the benchmark's maximum block size. 

    \item {\bf A Stronger Benchmark with Fractional Allocations:}
    One of the things we get  ``for free'' when we relaxed the maximum block size restriction
    to an average block size restriction, is that EIP-1559 turns out to also compete with any 
    {\em fractional} allocation that is allowed to arbitrarily split transactions among blocks or even
    only fractionally allocate transactions
    (even though EIP-1559 does not do that, obviously).   
\end{enumerate}

\subsection{Lower Bounds}

Let us now look deeper at the losses as stated in the main theorem.  Not only will we show
that EIP-1559 indeed may suffer such losses, but we will also prove that such losses must be
suffered by two classes of algorithms of
which EIP-1559 is a member: a loss for
the general class of online algorithms and an additional loss for a subclass of online algorithms 
that we term ``price-based algorithms'' that
operate by setting a price for each block that only depends on on-chain 
information from previous blocks (see 
\cref{def:price-based} for a formal definition).

\subsubsection{Maximum Block size}

The theorem requires that $c>2$, i.e.
that the 
EIP-1559 algorithm is allowed
{\em maximum} block size that is
more than twice the original block size.  
We prove such a blowup
in the {\em maximum} allowed block size
is required by {\em any online algorithm} that obtains a welfare guarantee like ours. 

The lower bound is proven even if $\Gamma$ and $\Delta$ can grow with $T$, as long as they grow as $o(T)$.
{Additionally, they hold even if the welfare approximation needs to be guaranteed only for one, known value $T$ (that is large enough).}

\begin{proposition} \label{prop:intro-lb-block}
   Any online algorithm that produces schedules with maximum block size $c \cdot B$, 
   for some $c<2$ {(with $c\geq 1$)} and average 
   block size limit $B$ with slackness $\Delta(T) = o(T)$, 
   even with an extension $\Gamma(T)=o(T)$, 
   must lose at least {$\delta_0=\Omega(2-c)$} 
   fraction of social welfare relative to schedules with 
   worst-case block size limit $B$ 
   (even when the values of all transactions are in $\{1,2\}$). 
\end{proposition}

Of special interest is the case $c=2$ which is the actual parameter used in Ethereum. 
While we do present an online algorithm (that can use information about non-scheduled transactions) with $c=2$ that loses no social welfare,
we show (\cref{sec:EIP-LB})
that for any horizon $T$ there is a possible input sequence where 
EIP-1559 with Ethereum's choice of parameter, $c=2$, does {\em not} compete with the optimum.  
We do note however that this is due to very large transactions
and, as mentioned above, if we limit the maximum transaction size by $q_{max}/B < 1$,
then a positive result emerges also for $c=2$.

\subsubsection{Dependence on Range of per-unit Values}

The theorem suffers a double loss: beyond 
allowing the algorithm $\Gamma$ more steps as well as having $\Delta$ slackness, 
we lose a fraction $\delta$ of welfare.  
Furthermore, both the number of extra steps $\Gamma$ and the slackness parameter $\Delta$
increase as a function of the range of possible per-unit values of transactions.  While
we exhibit an \emph{online algorithm} with zero loss of welfare and $\Delta=1, \Gamma=1$, we prove
that the dependence of $\Gamma$ or $\Delta$ on
the range of per-unit values is unavoidable for {\em price-based algorithms with good welfare guarantee}.

\begin{proposition}\label{prop:intro-value-range-lb}
Fix any price-based algorithm with average block size $B$, slackness $\Delta$, and extension $\Gamma$. 
There exists an instance with per-unit values in the range $[1,H]$ for which the algorithm does not obtain a $H^{1/4^{(\Gamma+\Delta)}}$  
fraction  of the optimal social welfare of a schedule with worst-case block size limit $B$ (with no slackness or extension).
Thus, if a price-based algorithm guarantees a constant fraction of this optimum, then $\Gamma+\Delta$ must grow as $\Omega(\log \log H)$.
\end{proposition}

Note that there is an exponential gap between our lower bound (which grows at rate $\Omega(\log \log H)$) and the upper bound (that is only singly-logarithmic in the per-unit value range).  
While we do show that the EIP-1559 algorithm may indeed require such a singly-logarithmic 
additional number of blocks in order to compete with the optimum, closing the gap
for general price-based algorithms remains open.  Especially intriguing is the
possibility of designing a useful ``variant'' of EIP-1559 with 
$\Delta+\Gamma$ that only grows at a rate double-logarithmic in the per-unit value range.

\subsection{Model Extensions}

Given our basic positive result of EIP-1559's near-optimality, it is tempting to 
generalize this to natural extensions of the basic model.  Specifically, we look at two tempting
extensions: to ``partially patient'' bidders and to ``multi-dimensional fees''.  In both
cases we prove the impossibility of an extension,
and show that every online algorithm must lose at least some fixed constant fraction $\delta_0>0$
of the optimal social welfare.  We leave open the question of whether one can recover at
least some, smaller, constant fraction of optimal social welfare.

\subsubsection{Partially Patient Bidders}\label{sec:intro-imp}

All our results so far assumed ``patient bidders'', i.e. where a transaction's value
for its user remains the same over time.  The opposite assumption of ``impatient
bidders'' assumes that a transaction {\em must} be scheduled in the immediate block
or it loses all value for its user.  A more general model would capture some sensitivity of the value of a transaction to the time of its execution, where
the value decreases with time.  The most common model for such
time-dependent value \cite{HLM21, PS24, GY24} would have a discount rate: 
A transaction with value $v$ and {\em discount factor} $\rho$ 
(where $0 \le \rho < 1$) 
has value for its user of 
$v \cdot (1-\rho)^{t_e-t_a}$, where $t_a$ is its arrival time and $t_e$ its execution time.

Thus the fully patient model corresponds to 
discount factor $\rho=0$, while the 
impatient model corresponds to $\rho\rightarrow 1$.
As the EIP-1559 algorithm does not allow its input to specify any 
discount factor, it is not hard to observe that if we just run EIP-1559 
on partially patient bidders then, since it does not distinguish between ``new'' transactions and
old ones that already lost most of their value, it cannot produce highly efficient results.  But perhaps if we just let the algorithm take into account, at every block, the
{\em current} value of the transaction, then we regain efficiency?  
I.e., suppose that at block $t$,
a pending transaction $i$ that has already arrived is considered for allocation if and only if
$v_i \cdot (1-\rho_i)^{t-t_i} \ge p_t$?  
Unfortunately, we get a negative answer and,
in fact, a lower bound for any online algorithm.  
Significantly, this lower bound
holds for a mix of patient bidders with either impatient bidders or partially patient bidders with arbitrarily low discount factors.

\begin{proposition}\label{prop:imp}
    Fix an online algorithm with average block size $B$, 
    slackness $\Delta$,
    and extension $\Gamma$, where 
    $\Delta(T)+\Gamma(T)=o(T)$. 
 For every {maximum} 
 discount rate ${\rho_{max}}>0$ 
    there exist a
    time horizon $T$, and an input sequence where
    every transaction $i$ has a discount rate 
    $\rho_i \in\{0, {\rho_{max}}\} $ 
    and a value in $v_i\in \{1,2\}$, 
 for which    the algorithm loses 
    at least a fraction $\delta_0=1/20 -o(1)$ of welfare relative to the optimal
    schedule with worst-case block size limit $B$.   
\end{proposition}

This impossibility result applies to scenarios where different transactions
may have different discount rates.  We do conjecture that the ``modified''
EIP-1559 is near optimal if the discount rate is global and shared by all
transactions.\footnote{Likely suffering an additional loss that behaves like $(1-\rho)^{O(\Delta+\Gamma)}$ {when all transactions have discount rate $\rho$}.}  We do show however, that if we  
replace the model of ``discount factor'' with a model where transactions
have a ``patience level'' where they can be scheduled within $p$ steps
of their arrival time without losing any value, but afterwards lose all value, then again, online algorithms must lose a constant fraction of
welfare, even if all transactions have the same patience level.

\subsubsection{Multi-dimensional Fees}

In our basic model each transaction has a ``single-dimensional'' size $q_i$.  More
generally one may consider a model where there are $m$ different resources, each
block has a size limit $B_j$ for every resource $j$ and each transaction $i$ uses
the amount $q_{ij}$ of each resource $j$ (see e.g. \cite{MDBFM23,ADM24}).  
The introduction of the ``blob'' resource to Ethereum
\cite{4844} is a step in this direction.
The natural generalization of the slackness condition to the multi-dimensional case
would require that every $T$ consecutive blocks use at most $(T+\Delta) \cdot B_j$
amount of each resource $j$.  
While one may hope to extend the near-optimality result to such a multi-dimensional model,
it turns out any online algorithm must lose some constant fraction of 
social
welfare in the multi-dimensional case.

\begin{proposition}
If there are at least three resources then any online algorithm with average block
size $B_j$, slackness $\Delta(T)=o(T)$ 
for each resource $j$ and given an extension $\Gamma(T)=o(T)$,
must lose at least a constant fraction 
$\delta_0=1/6-o(1)$
of social welfare on some input sequence.
\end{proposition}

{After the initial version of this paper was
published, \cite{BeN25} extended the lower bound also for the case of two resources, {and additionally,} 
gave a constant factor approximation for the general case.}

\subsection{Strategic Points of View}

Our analysis in this paper is purely algorithmic, so our results can be viewed as proving near-optimality in terms of the \emph{declared} social welfare: the welfare with respect to the values given as 
input by the bidders to the algorithm (rather than the true values). 
Regarding the operators (miners, validators), our
analysis assumed adversarial behavior, 
subject to the weak constraint that they must
schedule a maximal by inclusion subset of admissible transactions.  
While we do not formally make any strategic analysis,
we do want to shortly mention some strategic points of view.

From the point of view of the bidders, the known
incentive compatibility of EIP-1559 is only for impatient bidders\footnote{Even for impatient bidders, 
the incentive compatibility of EIP-1559
does not hold in the case of a rapid demand increase as then tips kick in and winners pay the tips they offer.}, while patient bidders may
in fact shade their bids profitably.  So can we expect
truthful bidding from patient bidders?  Well, practically, bidders
may simply act myopically -- in our case, truthfully -- due to the complexity of predicting the future, and their bounded rationality.
In our case, such behavior is also justified by risk aversion:
Bidders that are sufficiently 
risk-averse with respect to the future
will bid truthfully, taking the first available slot at a profit rather
than risk never being scheduled due to possible future increase in demand.
An appealing strategic analysis of patient bidders' behavior that
includes a reasonable model
of future uncertainty is of course 
of much interest, and is left as a future challenge.\footnote{
It is
interesting to note that if in equilibrium the declared values 
are monotone in 
the true values -- as may be expected in many models -- 
then optimizing for the former is closely related to optimizing for the latter. In that case, our algorithmic results may carry over to the
strategic setting in the model in question.
A precise statement to such a result would obviously depend on many details of the model.} 

From the point of view of the operators  
that assemble the blocks, our analysis undertook an adversarial point of
view. Thus, it should hold under a very general class of strategic behaviors of the operators, including any reasonable model of myopic operators
as studied so far in the literature.  
The only required assumption is that
scheduling a set of transactions is not less beneficial
to the operator than scheduling only a subset of them.

\section{Model and Notations}\label{sec:model}

\subsection{Blocks, Schedules, and Social Welfare}

We consider transactions that are processed in blocks. At each time  $t\in \{1,2,\ldots \}$ 
one block is processed. 
We consider 
{
a set $\trans$ of transactions that arrive online over time,
at each time $t\in \mathbb{N}$ a finite set of transactions arrives.
Each \emph{transaction} $i\in \trans$ has an arrival time $t_i\in \mathbb{N}$, size $q_i>0$, 
and value per-unit size of $v_i>0$ (thus its total value if executed is $q_i \cdot v_i$). 
Transactions are patient, i.e. may be scheduled at time $t_i$ or at any time afterwards, without loss in value.}
We use $q^T_{max}=\max_{i\in \trans} q_i$ to denote the maximal size of any transaction arriving up to time $T$, and use 
$v^T_{min}=\min_{i\in \trans} v_i$ and
$v^T_{max}=\max_{i\in \trans} v_i$ to denote the minimal and maximal per-unit value of any transaction arriving up to time $T$, respectively.
When $T$ is clear from the context we drop it from these notations and write $q_{max}, v_{min}$ and $v_{max}$ (instead of $q^T_{max}, v^T_{min}$ and $v^T_{max}$).

\begin{definition}
Given a set of transactions $\trans$ (where finitely many arrive at each time step) 
a \emph{schedule} $S$ is an assignment of {transactions} to times,  
such that {a transaction cannot be assigned before it arrives.} 
Formally, for each $i\in \trans$ and $t\in \mathbb{N}$, a \emph{schedule} $S=\{x_{i,t}\}_{i\in \trans, t\in \mathbb{N}}$ {where} 
$x_{i,t}\in \{0,1\}$ specifies, for each $i\in \trans, t\in \mathbb{N}$, whether transaction $i$ was scheduled at time $t$, under the constraints that for every transaction $i\in\trans$:
\begin{itemize}
    \item the transaction is not scheduled before it arrives: $t<t_i$ implies that $x_{i,t}=0$. 
    \item the transaction is assigned no more than one slot: 
$\sum_{t\in \mathbb{N}} x_{i,t}\leq 1$.
\end{itemize}  
\end{definition}

\begin{definition}
A \emph{fractional schedule} $S=\{x_{i,t}\}_{i\in \trans, t\in \mathbb{N}}$ is a schedule satisfying the exact same two constraints, 
but 
relaxing the integrality requirement on $x_{i,t}$, allowing for any fraction $x_{i,t}\in [0,1]$.
\end{definition}

{
\begin{definition}
    We use $Q_t(S)=\sum_{i\in \trans} x_{i,t} \cdot q_i$ to denote the total capacity used by the schedule $S$ at time $t$. When the schedule $S$ is clear from the context we
    may omit it in the notation and denote $Q_t= Q_t(S)$.  We say that a schedule has
    \emph{maximum block size $B$} if for all $t \ge 1$ we have that $Q_t \le B$.  
\end{definition}

}
\begin{definition}
We denote the set of transactions in $\trans$ with per-unit value at least $\theta$ by $\trans(\theta)$.
For a schedule $S=\{x_{i,t}\}_{i\in \trans, t\in \mathbb{N}}$ we use $Q_t(S,\theta)= \sum_{i\in \trans(\theta)} x_{i,t} \cdot q_i$ to denote the total capacity of transactions with per-unit  value at least $\theta$ that are scheduled at time $t$. We use $Q_{[1,T]}(S,\theta)=\sum_{t=1}^T Q_t(S,\theta)$ to denote the total capacity of transactions with per-unit value at least $\theta$ that are scheduled at times $1,2,\ldots,T$.
\end{definition}

Note that for brevity we have omitted the implied underlying set of transactions $\trans$ from the notations of $Q_t(S,\theta)$ and $Q_t(S)$. {We do the same in the following definition of social welfare.} 

\begin{definition}
The \emph{social welfare} of a schedule $S$ of transactions in $\trans$ up to time $T$, is the sum of the values of the transactions that were scheduled up to $T$: $SW(S, [1,T])=\sum_{i\in \trans,t\in [T]} x_{i,t} \cdot q_i\cdot v_i$.

\end{definition}

The social welfare of a  schedule may be easily expressed using the quantity parameters:

\begin{lemma} \label{lem:intQ} For any $T$ and any schedule $S$ it holds that 
    $$SW(S,[1,T]) = \int_0^\infty Q_{[1,T]}(S,\theta) d\theta$$
\end{lemma}

\begin{proof}
    By definition $Q_{[1,T]}(S,\theta)= \sum_{i\in \trans(\theta), t\in [T]} x_{i,t} \cdot q_i = \sum_{i\in \trans, t\in [T]} x_{i,t} \cdot q_i \cdot \mathbbm{1}_{v_i \ge \theta}$ 
    (notice that the summation over $i$ in the first summation is only over transactions with per-unit value of at least $\theta$, while in the second summation it is over all transactions).
    Integrating by $\theta$ we have $\int_0^\infty Q_{[1,T]}(S,\theta) d\theta =
    \sum_{i\in \trans, t\in [T]} x_{i,t} \cdot q_i \cdot \int_0^\infty \mathbbm{1}_{v_i \ge \theta} (\theta) d \theta = 
    \sum_{i\in \trans, t\in [T]} x_{i,t} \cdot q_i \cdot v_i$.  Now notice that this is exactly
    the definition of $SW(S,[1,T])$.
\end{proof}

\begin{definition} 
A schedule is said to have \emph{average block size limit $B$ with slackness $\Delta:\mathbb{N}\rightarrow \mathbb{N}$} if
for any number of time steps $k$, the total size of transactions during any $k$ consecutive blocks 
is at most $(k+\Delta(k)) \cdot B$.  
I.e. for any $t_0 \le t_1$ we have that  $\sum_{t=t_0}^{t_1} Q_t(S) \le (t_1-t_0+1+\Delta(t_1-t_0+1)) \cdot B$.
{When $\Delta$ is a constant function (for constant $Z$ it holds that $\Delta(k)=Z$), we abuse notation and simply write $\Delta=Z$.}
\end{definition}

\subsection{Online Scheduling Algorithms}

In this paper we consider \emph{online} scheduling algorithms. 
{While (offline) scheduling algorithms map all information about the transactions to a schedule, online algorithms are constrained to} 
assemble transactions into a sequence of blocks using only information about transactions arriving so far.

\begin{definition}
    A scheduling algorithm  is called \emph{online}  if its allocation at any time $t$ is determined only as a function of the information known by time $t$ (inclusive), i.e.
    on information from {every  transaction} $i$ with $t_i \le t$.
\end{definition}
    {We say that a scheduling algorithm has \emph{maximum block size $B$} if on every input it produces a schedule with maximum block size $B$.}

{It is perhaps useful to present the EIP-1559 
algorithm by first
defining a  sub-class of online
algorithms, termed ``price-based''\footnote{This class is conceptually
similar to the 
class considered in
\cite{FMPM21} but we take an adversarial
abstraction of the tip mechanism rather than their randomized choice.  This whole class shares EIP-1559's incentive properties in the impatient case.},  which the \EIP algorithm is a member of.
}

\begin{definition}\label{def:price-based}
    An online scheduling algorithm  is called \emph{price-based}, if at every time $t$ it sets a per-unit  price $p_t$ and maximum size $B_t$ based only on the parameters of the transactions that were executed at previous times (i.e. every 
    transaction $i$ such that $x_{it'}=1$ for $t'<t$).  The transactions that are scheduled at time $t$ are chosen as follows: look at all yet unscheduled transactions $i$ that have already arrived by time $t$ ($t_i \le t$) and are willing to pay $p_t$ per unit ($v_i \ge p_t$).  From these, an {\em adversary} chooses a maximal (by inclusion) subset of transactions of total size at most $B_t$.
\end{definition}

{Thus price-based scheduling algorithms
are only allowed to use ``on chain'' information (i.e. 
from transactions that were scheduled before time $t$) and their
allocation decisions must be simply by price per unit.
This is in contrast to the more general class of {\em online}
algorithms that may take decisions based on all information known by decision time, and use that
information in an arbitrary way.}

We next define a family of ``\EIP algorithms" - these are price-based algorithms that are parameterized by several parameters: 1) target block size $B>0$, 2) maximum block size $c\cdot B$ for some constant $c> 1$, 3) a price-update parameter $\eta>0$, 4) minimum price $p_{min}\geq 0$, and 5) an initial price $p_1$ for the first round, satisfying $p_1\geq p_{min}$. 

\begin{definition}
   The \emph{\EIP algorithm} with parameters $(B, c, \eta, p_{min}, p_1)$
   {such that $B>0,  c> 1, \eta>0, p_1\geq p_{min}>0$,}
   is a price-based algorithm  where $B_t=c\cdot B$ for every time $t$ {(so it has maximal block size $c\cdot B$)}, 
   and  where the per-unit price $p_t$  
   is computed iteratively as follows:
   {the price at time $t=1$ is $p_1$, and for any $t>1$ it holds that} {$p_{t+1} = \max \left\{p_{min}, p_t\cdot e^{\eta\cdot (Q_t-B)/B}\right\}$}, where $Q_t$ is the total size of the transactions scheduled at time $t$.
\end{definition}

Observe that the price goes up when $Q_t>B$, and it goes down when $Q_t<B$. The step size $\eta$ controls the rate in which the price is updated.

The actual EIP-1559 algorithm {with Ethereum's choice of parameters} is such an algorithm with target block size {18} 
Mega gas, maximum block size twice as large ($c=2$), price-update parameter $\eta=1/8$, $p_1 = 10^{-9}$ ETH and $p_{min}=10^{-18}$ ETH.\footnote{
The minimum price is implicit since all price calculations are done in integer 
multiples of Wei ($=10^{-18}$ ETH) and the first block's price, when moving to EIP-1559 was
set to one Gwei ($=10^{-9}$ ETH).}
We also remark that Ethereum's EIP-1559 algorithm approximates the price update multiplier $ e^{\eta\cdot (Q_t-B)/B}$ by  $1+{\eta\cdot (Q_t-B)/B}$ (using the first term of the Taylor expansion, $e^{x}\approx 1+x$).  Finally, when there is more demand at the current block's price than the maximum block size, then in Ethereum's EIP-1559 algorithm, the transactions taken are chosen according to their offered ``tip''
rather than adversarially.  From an  algorithmic (non-strategic) point of view, as we take here, since the tips may be unrelated to the actual values of the transactions, 
they {may}
indeed result in any adversarial schedule. 

{It is generally known that the \EIP algorithm yields blocks whose {\em average} size
approaches the target size $B$.  We present the formal statement and proof in our setting and notation.}  

\begin{proposition}\label{prop:EIP-slack}
The \EIP algorithm with parameters $(B, c, \eta, p_{min},p_1)$ 
has average block size limit $B$ with slackness $\Delta = \frac{1}{\eta}\cdot \ln \frac{v_{max}}{p_{min}}+ {(c-1)}$.
\end{proposition}

\begin{proof}
   Fix any $k\in \mathbb{N}$, and consider a sequence of $k$ consecutive blocks $[i,j]$ of length $k\geq 1$ (i.e. $j-i=k-1$). If  $\sum_{t=i}^j Q_t=0$ the claim trivially holds. Otherwise, let $i'\geq i$ be the lowest index such that $Q_{i'}>0$. Note that $i'\leq j$ and that
   $\sum_{t=i}^j Q_t=\sum_{t=i'}^j Q_t$.
   
   Since the prices are updated multiplicatively, we have that  
    \begin{equation*}
    \frac{p_{j+1}}{p_{i'}} =\prod_{t=i'}^{j} \frac{p_{t+1}}{p_t} = 
    \prod_{t=i'}^{j} \max \left\{\frac{p_{min}}{p_t}, e^{\eta\cdot (Q_t-B)/B}\right\}
    \geq e^{\eta\cdot (\sum_{t=i'}^j Q_t-k\cdot B)/B}
   \end{equation*}
   It holds that $p_{i'}\geq p_{min}$. Additionally, $p_{j+1}\leq  v_{max} \cdot e^{\eta\cdot (c-1)}$ by the following lemma and the fact that $Q_{i'}>0$ for $i'\leq j< j+1$:
   \begin{claim}
       If for time $t$ it holds that $Q_t>0$ then for any $k>t$ it holds that  
       $p_{k}\leq  v_{max} \cdot e^{\eta\cdot (c-1)}$.       
   \end{claim}
   \begin{proof}
       We prove the claim by induction. Assume the claim was true up to time $k-1\geq t$, so $p_{k-1}\leq  v_{max} \cdot e^{\eta\cdot (c-1)}$. If $p_k\leq p_{k-1}$ then $p_k\leq p_{k-1}\leq  v_{max} \cdot e^{\eta\cdot (c-1)}$. Otherwise  $p_k> p_{k-1}$, implying that demand at price $p_{k-1}$ was positive, so $p_{k-1}\leq v_{max}$. In that case, the  price can only rise above $v_{max}$ by a factor of at most $e^{\eta\cdot (Q_{k-1}-B)/B}\leq e^{\eta\cdot (c\cdot B-B)/B}= e^{\eta\cdot (c-1)}$ 
   since $Q_{k-1}\leq c\cdot B$. 
   \end{proof}
   
   We conclude that $\frac{p_{j+1}}{p_{i'}} \le \frac{v_{max}\cdot e^{\eta\cdot (c-1)}}{p_{min}}$. Combining the two inequalities we get 
   \begin{equation*}
    \frac{v_{max}\cdot e^{\eta\cdot (c-1)}}{p_{min}}\geq \frac{p_{j+1}}{p_{i'}} 
    \geq e^{\eta\cdot (\sum_{t=i'}^j Q_t-(j-i'+1)\cdot B)/B}
   \end{equation*}
   taking natural logs 
   \begin{equation*}
    \ln \frac{v_{max}}{p_{min}}+ {\eta\cdot (c-1)} 
    \geq {\frac{\eta}{B}\cdot \left(\sum_{t=i'}^j Q_t-(j-i'+1)\cdot B\right)}
   \end{equation*}
   thus
   \begin{equation*}
    B\left((j-i'+1)+  \frac{1}{\eta}\cdot \ln \frac{v_{max}}{p_{min}}+ {(c-1)}\right) 
    \geq \sum_{t=i'}^j Q_t = \sum_{t=i}^j Q_t
   \end{equation*}
   As $i'\geq i$  we have  $j-i'+1\leq j-i+1 = k$  and thus $B\cdot k \geq B\cdot (j-i'+1)$. 
   Thus, the  \EIP algorithm with parameters $(B, c, \eta, p_{min},p_1)$  has average block size limit $B$ with slackness 
    $\Delta = \frac{1}{\eta}\cdot \ln \frac{v_{max}}{p_{min}}+ {(c-1)}$.
\end{proof}

\section{\EIP has High Welfare for Patient Bidders}\label{sec:main}

In this section we prove our main result, showing that when bidders are patient, the \EIP algorithm produces schedules of near-optimal welfare, provided  it is given a mild resource augmentation (that does not increase with the time horizon).  {As a warm up we first analyze a Greedy Online Algorithm (\cref{sec:greedy}), illustrating some of the  techniques we use in  proving our main result. We then present our main result (\cref{sec:main-result}); its proof is based on a central lemma which we prove in \cref{sec:main-lemma}}.

\subsection{Warm up: a Greedy Online Algorithm} \label{sec:greedy}

Before diving into the main theorem, we can get some intuition by proving that a simple
online greedy algorithm with slightly relaxed constraints, competes with the optimum.  
Specifically, we show that a simple online greedy algorithm with $c=2$, slackness of $1$ and extension of $1$, has welfare at least as high as the optimum.
This algorithm can be viewed as setting prices at each time, but these prices are allowed to depend on all input arriving so far.
To some extent it is possible to view the EIP-1559 algorithm as attempting to mimic this online algorithm
by gradually approximating its prices.  This gradual approximation will take extra steps, require bounds
on the per-unit valuations, and incur some losses which the
simple greedy algorithm does not suffer from (as, unlike the EIP-1559 algorithm, it has full access to all parameters of the transactions that have already arrived).

Consider any time horizon $T$ (that need not be known by the algorithm). 
For any set of arriving transactions, we compare the welfare of the algorithm to the welfare of the optimal schedule over $T$ steps, where the optimal schedule is limited to have maximal block size $B$ at each time step (and thus we can assume that there are only transactions of size at most $B$).
At each time $t=1,2, \ldots,T$ our algorithm will greedily schedule pending transactions with the highest per-unit value, until the first point in which total size of the transactions  
scheduled by time $t$ is at least $t \cdot B$ (the case there are not enough pending transactions is an edge case that we may ignore, by pretending to have an unlimited amount of 0 value transactions). As transactions are of size at most $B$, we have that $t\cdot B \le \sum_{i=1}^t q_t< (t+1)\cdot B$, where $q_t$ is the total
size of transactions scheduled for block $t$.

Let us denote  by $p_t$ the value (per unit) of the last (lowest value per unit)  transaction scheduled at time $t$.  

We first observe that for any sequence of $Z$ blocks we have:
$(Z-1)\cdot B \le \sum_{i=t}^{t+Z-1} q_t < (Z+1)\cdot B$, {and thus} this algorithm has average 
block size $B$ with slackness 1.  
This also implies that it has maximum 
block size $2\cdot B$.  

We next observe that when given an extension of one extra block, this greedy online algorithm always achieves social welfare {at least} as high as the optimum schedule of maximal block size $B$.

We denote {the schedule produced by this algorithm} by $ALG$ and the optimal {(welfare-maximizing up to time $T$)} fractional schedule {of maximal block size $B$} by $OPT$.

\begin{proposition}\label{warmup} For any $T$ it holds that 
    $SW(ALG, [1,T+1]) \ge SW(OPT, [1,T])$.
\end{proposition}

\begin{proof}

In order to analyze its performance, we make the following claim.  Fix an arbitrary positive number $\theta$ and denote by $Q_{[1,t]}(ALG,\theta)$  the total size of transactions whose value (per unit)
is at least $\theta$ that were scheduled by the algorithm up to (and including) time 
$t$.  Similarly, denote by $Q_{[1,t]}(OPT,\theta)$ the total size of transactions whose value (per unit) is at least $\theta$ 
that were scheduled by the optimal (even fractional) schedule
that has maximum block size $B$ up to (and including) time $t$.

\begin{lemma}
    For every $t$ and $\theta$ we have that $Q_{[1,t+1]}(ALG,\theta) \ge Q_{[1,t]}(OPT,\theta)$.
\end{lemma}

\begin{proof}
The proof of this lemma proceeds by splitting the algorithm's run into two phases.  Let $\alpha_\theta$ 
be the last (largest) time $t$ such that $p_t < \theta$ (If this never happens, then 
we take $\alpha_\theta=0$).  
By the end of the first phase,
i.e. at time $\alpha_\theta$,
there are {\em no} pending transactions 
whose value per unit is at least $\theta$ (since these should have been scheduled before the transaction
whose value per unit is $p_{\alpha_\theta}$). Thus, during the first phase, 
the total size of   transactions with value at least $\theta$ that ALG schedules is at least as high as the total scheduled by OPT, that is, $Q_{[1,\alpha_\theta]}(ALG,\theta) \ge Q_{[1,\alpha_\theta]}(OPT,\theta)$. If $\alpha_\theta\geq T$, the claim follows. Otherwise, $\alpha_\theta<T$.
Now we look at the second phase, i.e. during times steps $\alpha_\theta+1,\ldots,T, T+1$.  
For any time $t$ in this time range it holds that $p_t \ge \theta$, and thus every transaction
scheduled by the algorithm in that range has value per unit of at least $\theta$.
Since, as mentioned above, in every $Z$ consecutive blocks the algorithm schedules transactions of 
size at least $(Z-1)\cdot B$, we have that during the second stage 
(including time $T+1$) our algorithm scheduled transactions of 
size at least 
$((T+1)-\alpha_\theta-1)\cdot B$, 
while $OPT$ can schedule transactions of 
size at most
$(T-\alpha_\theta)\cdot B$ 
during the range $\alpha_\theta+1, \ldots, T-1, T$ (i.e without the 
extra step at time $(T+1)$). 
It follows that in the second phase,  the total size of   transactions with value at least $\theta$ that ALG schedules (with the extra time step)  is at least as high as the total scheduled by OPT, that is, $Q_{[\alpha_\theta+1, T+1]}(ALG,\theta) \ge (T-\alpha_\theta)\cdot B \ge Q_{[\alpha_\theta+1,T]}(OPT,\theta)$. 
\end{proof}

The proposition can be directly deduced from this lemma using \cref{lem:intQ}
that states that the social welfare achieved by $ALG$ for time $1$ till $t+1$  
can be expressed as $\int_0^\infty Q_{[1,t+1]}(ALG,\theta) d\theta$,  and
similarly, the  social welfare achieved by $OPT$ for time $1$ till $t$  
can be expressed as $\int_0^\infty Q_{[1,t]}(OPT,\theta) d\theta$.
\end{proof}

{\bf Remark:} This analysis did not assume any bound on the maximum transaction size beyond being bounded by the target block size $q_{max} \le B$ (since
we are competing with schedules with maximum block size $B$).  If we have a bound on the maximum transaction size then the maximum block size
of this algorithm is in fact sharper: $c=1+q_{max}/B$.

\subsection{Our Main Result}\label{sec:main-result}

Our main result is that \EIP algorithms have the attractive property that for patient  bidders (with some bounds on the per-unit valuations), for any time horizon,   
they give close to optimal welfare with respect to the benchmark of all schedules 
that are restricted to maximum block size $B$ and are running for slightly fewer steps.

This may be viewed as surprising since the algorithm is both restricted to online decisions, and furthermore, is restricted to be price-based, and yet competes with all offline clairvoyant schedules. On the other hand, it is not limited by a worst-case block size limit $B$, but rather block size limit $B$ holds only on the average with some slackness $\Delta$ (as shown in \cref{prop:EIP-slack}). 
Additionally, the algorithm gets an extension of additional $\Gamma$ time steps, for some appropriately chosen $\Gamma$ (that does \emph{not} grow with the time $T$).

{We use $ALG$ to denote the schedule produced by an \EIP algorithm. Recall that for schedule $S$ we use $SW(S,[1,t])$ the total social welfare of $S$ up to time $t$.}

\begin{btheorem}\label{thm:basic}
        Fix an \EIP algorithm with parameters $(B, c, \eta, p_{min},p_1)$ 
        for some $c>2$. 
        Consider any (adversarially chosen) sequence $\trans$ of input transactions where each transaction $i\in \trans$ has per-unit value $v_i$ satisfying
        $e^{\eta} \cdot p_{min} \le v_i \le v_{max}$ 
        for some 
        $v_{max}\ge p_1$. Let $OPT_B$ be an arbitrary schedule with maximum block size $B$.
        Then for any time $T$ it holds that 
          $$SW(ALG,[1,T+\Gamma])
    \ge (1-\delta) \cdot 
    SW(OPT_B,[1,T])
    $$ 
for
     $\delta = 1-e^{-\eta}\leq \eta$ and
     {any integer}
     $\Gamma {\geq 
     \frac{1}
     {\eta}}\cdot \ln\left(\frac{v_{max}}{p_{min}}\right) \cdot \max \{1,\frac{1}{c-2}\} + c$.  
\end{btheorem}

\cref{thm:basic} is a corollary of a more general and stronger result that we prove. First, instead of the benchmark being only integral allocations, we can take the benchmark to be all \emph{fractional} allocations (in which a transaction with per-unit value $v_i$ that is allocated a fraction $x_i\in [0,1]$ contributes $v_i\cdot q_i\cdot x_i$ to the value of the schedule). Second, we present tighter results for the case that $q_{max}<B$, replacing the assumption that $c>2$ with the assumption that $c>1+q_{max}/B$.
Third, instead of the benchmark having worst-case block size limit $B$, we allow it to only have \emph{average  block size limit $B$ with slackness $\Delta'$}. To allow for this extra slackness we will allow the algorithm an additional $\Delta'$ time steps (effectively increasing $\Gamma$ by $\Delta'$). 
Finally, we will allow the initial price $p_1$ to be arbitrary. When it is larger than $v_{max}$ it will imply that $\Gamma$ grows as $\ln \frac{p_1}{p_{min}}$. 

\begin{btheorem}\label{thm--ub--eip1559-optimal-frac}    
{Fix an \EIP algorithm with parameters $(B, c, \eta, p_{min},p_1)$. 
        Consider any (adversarially chosen) sequence $\trans$ of input transactions where each transaction $i\in \trans$  has size at most $q_{max}$ and has per-unit value $v_i$ satisfying
        $e^{\eta} \cdot p_{min} \le v_i \le v_{max}$. Assume that $c>1+\frac{q_{max}}{B}$.
        Let $OPT_{B,\Delta'}$ be an arbitrary \emph{fractional schedule with average block size limit $B$ and slackness $\Delta'$}.
        Then for any time $T$ it holds that 
          $$SW(ALG,[1,T+\Gamma])
    \ge (1-\delta) \cdot 
    SW(OPT_{B,\Delta'},[1,T])
    $$
     for $\delta = 1-e^{-\eta}\leq \eta$ and 
     {for any integer 
     {$\Gamma \geq \max \left\{  \frac{1}{\eta}\ln \left(\frac{p_{1}}{p_{min}}\right),   \frac{1}{\eta(c'-1)}\cdot \ln \left(\frac{v_{max}}{p_{min}}\right) +c-1  
+  \frac{c-2}{c'-1}\right\} +\Delta' $}, 
where  $c'=c-\frac{q_{max}}{B}$.}
}
\end{btheorem}
{We note that we do not assume any knowledge of $v_{max}$ by the algorithm (it is only used for analysis).}

{To see that \cref{thm:basic} follows from \cref{thm--ub--eip1559-optimal-frac}, observe that
when $q_{max}=B$ we have  
$c>1+\frac{q_{max}}{B}=2$.
Additionally, the benchmark of integral schedules with maximal block size $B$ is clearly weaker than the benchmark of fractional schedule with average block size limit $B$ and slackness $\Delta'\geq 0$. Finally, as $c=c'+1$, the assumption that $v_{max}\ge p_1$ implies that
$\frac{1}{\eta}\ln \left(\frac{p_{1}}{p_{min}}\right)\leq \frac{1}{\eta}\ln \left(\frac{v_{max}}{p_{min}}\right)$ and thus $\Gamma$ must be at least $c+  \max \left\{1,\frac{1}{c-2}\right\} \cdot
     \frac{1}{\eta}\cdot \ln \left(\frac{v_{max}}{p_{min}}\right)$ as 
$$
c+  \max \left\{1,\frac{1}{c-2}\right\} \cdot
     \frac{1}{\eta}\cdot \ln \left(\frac{v_{max}}{p_{min}}\right)=
c-1  
+  \frac{c-2}{(c-1)-1}
+  \max \left\{1,\frac{1}{c-2}\right\} \cdot
     \frac{1}{\eta}\cdot \ln \left(\frac{v_{max}}{p_{min}}\right) \geq \frac{1}{\eta}\ln \left(\frac{p_{1}}{p_{min}}\right)$$}

{We remark that if $p_{min}>v_{min}\cdot e^{-\eta}$ then the result still holds, but with respect to schedules 
that are only allowed to schedule transactions with value per unit  at least $p_{min}\cdot e^{\eta}$. That is, there is an additional additive loss in welfare that is bounded by the welfare of any schedule of transactions of value per unit less than $p_{min}\cdot e^{\eta}$.
}

The cornerstone of the proof of the main result is the next lemma which shows 
{that for any $\theta$, the \EIP algorithm schedules large enough size 
of transactions of per-unit value $\theta$, as long as it gets a long enough extension. The theorem will be easily deduced from this bound using 
\cref{lem:intQ}.
}
{For the next lemma, recall that $Q_{t}(S,\theta)$ denotes the total size of transactions with per-unit value of at least $\theta$ that were scheduled by $S$ at time $t$.
}

\begin{lemma}\label{lem:cover}
    Fix an \EIP algorithm with parameters $(B, c, \eta, p_{min},p_1)$. 
    Consider any (adversarially chosen)
    sequence of input transactions, 
    and assume that each transaction up to time $T$ has a  per-unit value at least $v_{min}$, and that $v_{min}\geq e^{\eta} \cdot p_{min} $.    
    {Assume also that $c>1+\frac{q_{max}}{B}$. Let $E$ denote the schedule of the \EIP algorithm on the input.}
For any fractional schedule $S$
with average block size limit $B$ and slackness $\Delta'$, 
for any 
$\theta \geq 0$ 
 it holds that

 \begin{equation}\label{eq:sizes}
     Q_{[1,T]}(S, \theta)= \sum_{t=1}^{T} Q_{t}(S,\theta) 
     \le \sum_{t=1}^{T+\Gamma} Q_{t}(E,\theta \cdot e^{-\eta}) 
     =Q_{[1,T+\Gamma]}(E, \theta \cdot e^{-\eta})
\end{equation}

{for any integer {$\Gamma \geq \max \left\{  \frac{1}{\eta}\ln \left(\frac{p_{1}}{p_{min}}\right),   \frac{1}{\eta(c'-1)}\cdot \ln \left(\frac{v_{max}}{p_{min}}\right) +c-1  
+  \frac{c-2}{c'-1}\right\} +\Delta' $},
where  $c'=c-\frac{q_{max}}{B}$.}
\end{lemma}

{
The theorem easily follows from the main lemma by integrating 
\cref{eq:sizes} over $\theta$, and then applying  \cref{lem:intQ}.
By integrating \cref{eq:sizes} we conclude that
\begin{equation}\label{eq:int-sizes}
\int_0^\infty 
Q_{[1,T]}(S, \theta) 
d\theta 
\le \int_0^\infty Q_{[1,T+\Gamma]}(E, \theta \cdot e^{-\eta})
d\theta
\end{equation}

By \cref{lem:intQ} the LHS of \cref{eq:int-sizes} is exactly $SW(S, [1,T])$.  
To evaluate the RHS of \cref{eq:int-sizes},
we make a linear change of variable (with our notation  
reusing $\theta$)
$\int_0^\infty 
Q_{[1,T+\Gamma]}(E, \theta \cdot e^{-\eta})
d \theta=
e^\eta \int_0^\infty Q_{[1,T+\Gamma]}(E, \theta )
d\theta$, and
we use 
\cref{lem:intQ} again to conclude that the RHS is exactly 
$e^\eta \cdot SW(E, [1,T+\Gamma])$,
as needed.}

\subsection{Proof of The Main Lemma}\label{sec:main-lemma}

{Before giving a formal proof of \cref{lem:cover}, let us give the intuition drawing
parallels with the online algorithm above.  Similarly to the previous
analysis we will look at the last point in time
$\alpha_\theta$ (to be simply called $\alpha$ below) where the price was less than $\theta$ and the quantity scheduled is at most $Q_t\leq c\cdot B-q_{max}$, meaning that every transaction with per-unit value at least the price was scheduled.  
As previously, up to that point, {the \EIP algorithm} 
scheduled everything with per-unit value of at least $\theta$, so it competes with {any schedule $S$ (and also with a fractional one).} 
From this point on there are always enough pending transactions that are willing to pay at least $\theta$ per-unit (quantity larger than $c\cdot B-q_{max}$), so 
EIP-1559 
{prices will increase rapidly until they are 
at least $\theta'=\theta\cdot e^{-\eta}$ 
(at a time to be called $\beta$ below),
from which point they 
can never decrease much below $\theta'$} (prices will always be at least $\theta'=\theta\cdot e^{-\eta}$, the minimum price when price decreases from $\theta$ in a single step). 
So after this point in time, {the \EIP algorithm} 
fills its quota with transactions of value per unit of at least $\theta\cdot e^{-\eta}$, 
and since its average block size {limit} is the same as that of $S$,
it competes
with it successfully given enough scheduling time. 
To compensate for the time it took for the price to rise up to at least $\theta'$ (the time between $\alpha$ and $\beta$), 
we allow the algorithm an extended time of extra $\Gamma$ steps to schedule transactions. Additionally, {we allow an extra $\Delta'$ steps to the algorithm, in order for it
to compete also with schedules that
have {\em average} block size limitations 
with slackness $\Delta'$.}
We suffer a loss of a $1-e^{-\eta}$ factor
due to the possible single price decrease below $\theta$.

We next present the proof of \cref{lem:cover}. 
\begin{proof}
We prove the claim for $\theta \in [v_{min}, v_{max}]$ since the claim is trivial for $\theta>v_{max}$, and immediate for $\theta<v_{min}$ from the claim for $\theta=v_{min}$.
Recall that $c'=c-\frac{q_{max}}{B}$. 
As 
$c'> 1$ it holds that $c\cdot B-q_{max} = c'\cdot B > B$. 
We denote 
$\theta'= \theta \cdot e^{-\eta}$ 
and observe that $\theta>\theta'> \theta \cdot (1-\eta)$ and that $\theta'=e^{-\eta}\cdot \theta\geq e^{-\eta}\cdot v_{min} \geq p_{min}$ as $\theta \geq  v_{min} \geq  e^{\eta}\cdot p_{min} $. 
We say that \emph{demand at price $\theta$ was met at time $t$} if every transaction $i$ such that $v_i\geq \theta$ and $t_i\leq t$ has been scheduled by the algorithm at or before time $t$. 
Observe that if at time $t$ it holds that $p_t\leq \theta $ and $Q_t\leq c\cdot B-q_{max} = c'\cdot B$, then every transaction $i$ such that $v_i\geq \theta$ and $t_i\leq t$ must have been scheduled by the algorithm (as any such unscheduled transaction can fit the block $t$ as $Q_t+q_i\leq c\cdot B-q_{max}+q_i\leq c\cdot B$). In such a case we say that \emph{demand at price $\theta$ was necessarily met at time $t$} (as all demand at price $\theta$ was scheduled for sure by time $t$).

Trivially, demand at any price $\theta$ was necessarily met at time $t=0$ (as no demand arrived yet).
Let $\alpha=\alpha_\theta\in [0,T+\Gamma]$ denote the latest time $t\in [0,T+\Gamma]$ such that demand at price $\theta$ was necessarily met at time $t$. We assume for convenience that $p_0=p_{min} \leq \theta$ and $Q_0=0$, so for $t=\alpha$ (even when $\alpha=0$) it holds that $p_{\alpha}\leq  \theta$ and $Q_{\alpha}\leq c'\cdot B$.

To simplify notation we use $S_t(\theta) = Q_{t}(S,\theta)$ and $E_t(\theta) = Q_{t}(E,\theta)$.
Let $Z_{t}(\theta)$ denote the total size of all transactions
that arrived by time $t$ with value at least $\theta$, i.e. {the size of transaction $i$ is counted towards $Z_{t}(\theta)$}
if $v_i\geq \theta$ and $t_i\leq t$. It holds that
\begin{equation*}
\begin{aligned}
 \sum_{t=1}^{\alpha} S_t(\theta)  \leq  Z_{\alpha}(\theta) 
=  \sum_{t=1}^{\alpha} E_t(\theta)
\leq  \sum_{t=1}^{\alpha} E_t(\theta') 
\end{aligned}
\end{equation*}
where the rightmost inequality is due to the fact that $E_t$ is a non-increasing function of $\theta$ and since $\theta\geq \theta'$. 

Thus, if $\alpha\geq T$ then our claim follows as 
\begin{equation*}
\begin{aligned}
\sum_{t=1}^{T} S_t(\theta) & \leq \sum_{t=1}^{\alpha} S_t(\theta)  
\leq  \sum_{t=1}^{\alpha} E_t(\theta') \leq  \sum_{t=1}^{T+\Gamma} E_t(\theta')
\end{aligned}
\end{equation*}

So  we are left with the case that 
$0\leq \alpha< T$. In that case we have 
\begin{equation*}
\begin{aligned}
 \sum_{t=1}^{T} S_t(\theta) =
 \sum_{t=1}^{\alpha} S_t(\theta) +\sum_{t=\alpha+1}^{T} S_t(\theta)   \leq  Z_{\alpha}(\theta) +\sum_{t=\alpha+1}^{T} S_t(\theta) 
\end{aligned}
\end{equation*}

As demand at price $\theta$ was necessarily met at  $\alpha$,  we know that $\sum_{t=1}^{\alpha} E_t(\theta')\geq \sum_{t=1}^{\alpha} E_t(\theta)
=  Z_{\alpha}(\theta) $ so
\begin{equation*}
\begin{aligned}
\sum_{t=1}^{T+\Gamma} E_t(\theta') = 
 \sum_{t=1}^{\alpha} E_t(\theta') +\sum_{t=\alpha+ 1}^{T+\Gamma} E_t(\theta') \geq  
  Z_{\alpha}(\theta)+\sum_{t=\alpha+ 1}^{T+\Gamma} E_t(\theta')
\end{aligned}
\end{equation*}
We conclude that to complete the proof we only need to show that $\sum_{t=\alpha+1}^{T} S_t(\theta) \leq \sum_{t=\alpha+ 1}^{T+\Gamma} E_t(\theta')$, knowing that $0\leq \alpha< T$ and that demand at price $\theta$ was {not} necessarily met at any time $t>\alpha$. 

Thus, in the rest of the proof we assume that demand at price $\theta$ was not necessarily met at any time $t\in [\alpha+1,T+\Gamma]$, 
and move to lower bound the total capacity scheduled by the \EIP algorithm between time $\alpha+1$ and $T+\Gamma$.
As demand at price $\theta$ was not necessarily met at any time $t\in [\alpha+1,T+\Gamma]$, then for $t$ in this range of times, whenever
$p_t\leq  \theta$ then  $Q_t> c'\cdot B\geq B$. 
We conclude that if $p_t\leq \theta$ then $p_{t+1} = p_t\cdot e^{\eta\cdot (Q_t-B)/B}>  p_t\cdot e^{\eta\cdot (c'-1)} $. 
Thus, if {$p_{\alpha}< \theta$}  then price will start rising till at some time $t\geq \alpha+1$ it will hold that $p_t\geq \theta\geq \theta'$.

Additionally, we claim that if for some $t>\alpha$ it holds that $p_t\geq \theta'$, then $p_{t+1}\geq \theta'$ (implying that for any $t'>t$ it also holds that 
$p_{t'}\geq \theta'$). Indeed, if $\theta> p_t\geq \theta'$ then $p_{t+1} = p_t\cdot e^{\eta\cdot (Q_t-B)/B}>  p_t\cdot e^{\eta\cdot (c'-1)}\geq p_t\geq \theta'$.
On the other hand, {when} $\theta\leq p_t$ then $p_{t+1} = p_t\cdot e^{\eta\cdot (Q_t-B)/B}\geq   
p_t\cdot e^{-\eta}\geq 
\theta\cdot e^{-\eta}\geq 
\theta'$. 
We conclude that once $p_t\geq \theta'$ it holds $p_{t'}\geq \theta'$ for any $t'>t$.

Now, we use $\beta\in [\alpha+1,T+\Gamma]$ to denote the first time the price was at least $\theta'$ (the existence of such a time is implied by the arguments below since $\Gamma$ is large enough).
So, if $\beta>\mbe{\alpha+} 1$ then it holds that 
$p_{\beta}\geq \theta'$ and $p_{\beta-1}< \theta'$ 
(note that for any time $t\in [\beta, T+\Gamma]$ it holds that $p_t\geq \theta'$). 

To complete the proof we only need to prove the following lemma:
\begin{lemma} It holds that 
\begin{equation*}
\sum_{t=\alpha+1}^{T+\Gamma} E_t(\theta') \geq 
\sum_{t=\beta}^{T+\Gamma} E_t(\theta') 
=\sum_{t=\beta}^{T+\Gamma} {Q_t} 
\geq  
B\cdot (T-\alpha +\Delta')\geq 
    \sum_{t=\alpha+1}^{T} S_t(\theta) 
\end{equation*}
\end{lemma}
\begin{proof}
    Observe that the leftmost inequality is trivial as $\beta\geq \alpha+1$. Let us now consider the rightmost inequality. It holds due to the average block size limitation constraint of $B$ with slackness $\Delta'$, which states that  $\sum_{t=\alpha+1}^{T} S_t(\theta) \leq (T-\alpha+\Delta') \cdot B$. The equality  $\sum_{t=\beta}^{T+\Gamma} E_t(\theta') = \sum_{t=\beta}^{T+\Gamma} Q_t  $ holds as for any time $t\in [\beta, T+\Gamma]$ it holds that $p_t\geq \theta'$, and thus all capacity is used to schedule transactions with value per unit  at least $\theta'$.
    
    So to complete the proof we only need to show that $\sum_{t=\beta}^{T+\Gamma} {Q_t} \geq  B\cdot (T-\alpha +\Delta')$.
For any time $t\in [\beta, T+\Gamma]$ it holds that $p_t\geq\theta'=e^{-\eta}\cdot \theta\geq e^{-\eta}\cdot v_{min} \geq p_{min}$
(and thus also $p_{t+1}\geq p_{min}$) so $p_{t+1}=p_t\cdot e^{\eta\cdot \frac{Q_t-B}{B}}$. 

Thus
   \begin{equation*}
    \frac{p_{{T+\Gamma+1}}}{p_{\beta}} = \prod_{t=\beta} ^{T+\Gamma} \frac{p_{t+1}}{p_t} 
    = \prod_{t=\beta}^{T+\Gamma} e^{\eta\cdot \frac{Q_t-B}{B}} = 
    e^{\eta\cdot \left(\sum_{t=\beta}^{T+\Gamma} \frac{Q_t-B}{B}\right)}
   \end{equation*}
So
   \begin{equation*}
    \frac{B}{\eta}\ln \left(\frac{p_{{T+\Gamma+1}}}{p_{\beta}}\right) = \left(\sum_{t=\beta}^{T+\Gamma} {Q_t}\right)-B\cdot \left(T+\Gamma-\beta+1\right)
   \end{equation*}
or equivalently,  
   \begin{equation}\label{Eq:sumQt}
   \sum_{t=\beta}^{T+\Gamma} {Q_t} 
   = B\cdot \left(T + 
    \frac{1}{\eta}\ln \left(\frac{p_{{T+\Gamma+1}}}{p_{\beta}}\right)  +  \Gamma-(\beta-1)\right)
   \end{equation}
So proving that $\sum_{t=\beta}^{T+\Gamma} {Q_t} \geq  B\cdot (T-\alpha +\Delta')$ is immediate from the inequality 
$$\frac{1}{\eta}\ln \left(\frac{p_{{T+\Gamma+1}}}{p_{\beta}}\right)  +  \Gamma-(\beta-1)\geq \Delta'-\alpha$$ which we prove in Lemma \ref{lem:technical-ineq}.

\end{proof}

\begin{lemma}\label{lem:technical-ineq}
    It holds that 
    \begin{equation*}
      \Gamma\geq \frac{1}{\eta}\ln \left(\frac{p_{\beta}}{p_{{T+\Gamma+1}}}\right)  + \Delta' +\beta -(\alpha+1) 
\end{equation*}
\end{lemma}
\begin{proof}
Let us first {get two special cases out of the way.  The first special case is} the case that $\beta=1$ and $\alpha=0$. In this case $p_{\beta}=p_1$
and $p_{{T+\Gamma+1}}\geq p_{min}$.
So
$$\Gamma\geq 
\Delta' +  \frac{1}{\eta}\ln \left(\frac{p_{1}}{p_{min}}\right)
\geq 
\Delta' +  \frac{1}{\eta}\ln \left(\frac{p_{\beta}}{p_{{T+\Gamma+1}}}\right)
$$
where the left inequality holds by the definition of $\Gamma$,
and the claim follows. Thus, in the rest of the proof we assume that it is not the case that $\beta=\alpha+1=1$.

{The second special case is
$\beta=\alpha+1>1$. In this case 
$$p_{\beta}=p_{\alpha+1}=p_{\alpha}\cdot e^{\eta\cdot (Q_\alpha-B)/B}
\leq \theta\cdot e^{\eta\cdot (c'-1)}$$
where the inequality holds since $p_\alpha\leq \theta$ and $Q_\alpha\leq c'B$.
As once the price gets to $\theta'$ it never drops below it, we have $p_{T+\Gamma+1}\geq \theta'$. We conclude that  
$$\frac{p_{\beta}}{p_{T+\Gamma+1}}  \leq     \frac{\theta\cdot e^{\eta\cdot (c'-1)}}{\theta'}= 
e^{\eta\cdot c'}$$
and therefore
\begin{equation*} 
   \frac{1}{\eta}\ln \left(\frac{p_{\beta}}{p_{{T+\Gamma+1}}}\right) \leq 
    \frac{1}{\eta}\ln (e^{\eta\cdot c'})=c'
\end{equation*}
As $\beta-(\alpha+1)=0$, we get
\begin{equation}\label{eq:cpDeltap}
\frac{1}{\eta}\ln \left(\frac{p_{\beta}}{p_{{T+\Gamma+1}}}\right) + \Delta' +\beta -(\alpha+1) \leq c'+\Delta'.
\end{equation}
By the assumption about $\Gamma$, 
\[
\Gamma \geq   \frac{1}{\eta(c'-1)}\cdot \ln \left(\frac{v_{max}}{p_{min}}\right) +c-1  
+  \frac{c-2}{c'-1} +\Delta'.
\]
Since $v_{max}\geq v_{min}\geq e^\eta p_{min}$, we have
\[
\frac{1}{\eta(c'-1)}\ln \left(\frac{v_{max}}{p_{min}}\right)\geq \frac{1}{c'-1}.
\]
Therefore
\[
\Gamma
\geq
\Delta' + c-1 + \frac{1}{c'-1} + \frac{c-2}{c'-1}
=
\Delta' + c-1 + \frac{c-1}{c'-1}.
\]
As $c'\leq c$, we have $c'-1\leq c-1$, and hence
\[
\frac{c-1}{c'-1}\geq 1.
\]
So
\[
\Gamma \geq \Delta' + c \geq \Delta' + c',
\]
and the claim follows by \cref{eq:cpDeltap}. 
}

{We now analyze the general case where $\beta>\alpha+1\geq 1$.}  Recall that $\beta\in [\alpha+1 ,T+\Gamma]$ is the first time such that $p_{\beta}\geq \theta'$.
Thus, as $\beta>\alpha+1$ it holds that $p_{\beta-1}< \theta'$ 
(note that we have shown that for any time $t\in [\beta, T+\Gamma]$ it holds that $p_t\geq \theta'$),
and as $p_{\alpha+1}\geq p_{min}$ we get 
   \begin{equation*}
    \frac{\theta'}{p_{min}}\geq \frac{p_{\beta-1}}{p_{\alpha+1}} 
    \geq e^{\eta\cdot (\sum_{t=\alpha+1}^{\beta-1} (Q_t-B))/B}
    \geq e^{\eta\cdot (\sum_{t=\alpha+ 1}^{\beta-1} (c'-1) )} 
    = e^{\eta\cdot {(\beta-\alpha-1)} (c'-1) } 
   \end{equation*}
Thus
   \begin{equation*}
    \ln \left(\frac{\theta'}{p_{min}}\right)\geq {\eta\cdot {(\beta-\alpha- 1)} (c'-1) } 
   \end{equation*}

and as $\eta>0$ and $ c'>1$ we get 
   \begin{equation} \label{eq:bound-beta}
\frac{1}{\eta(c'-1)}\cdot \ln \left(\frac{\theta'}{p_{min}}\right)\geq \beta-\alpha-1 
   \end{equation}

    As $\eta>0 $ and  $c>c'>1$ (so $c-2>-1)$ it holds 
    that $v_{max}\cdot e^{\eta\cdot (c-2)}> v_{max} \cdot e^{-\eta}
    \geq \theta\cdot e^{-\eta}\geq \theta'$, 
    by
    \cref{eq:bound-beta} 
   \begin{equation}\label{eq:bound-beta2} 
\frac{1}{\eta(c'-1)}\cdot \ln \left(\frac{{v_{max}\cdot e^{\eta\cdot (c-2)}}}{p_{min}}\right)\geq
\frac{1}{\eta(c'-1)}\cdot \ln \left(\frac{\theta'}{p_{min}}\right)\geq \beta-(\alpha+1) 
   \end{equation}

   Since $\beta>\mbe{\alpha+}1$ is the first time the price is at least $\theta'$, 
   the price was rising from time $\beta-1$, and 
   $p_{\beta-1}<\theta'= \theta \cdot e^{-\eta}$. In this case 
   $p_{\beta}$  is at most 
   $\theta'\cdot e^{\eta\cdot (c-1)}$. 
As once the price gets to $\theta'$ it never drops below it, we have $p_{T+\Gamma+1}\geq \theta'$. We conclude that  
$\frac{p_{\beta}}{p_{T+\Gamma+1}}  \leq     \frac{\theta'\cdot e^{\eta\cdot (c-1)}}{\theta'}= 
e^{\eta\cdot (c-1)}$  and therefore 
   \begin{equation*} 
   \frac{1}{\eta}\ln \left(\frac{p_{\beta}}{p_{{T+\Gamma+1}}}\right) \leq 
    \frac{1}{\eta}\ln (e^{\eta\cdot (c-1)})=c-1
   \end{equation*}
Combining this with
\cref{eq:bound-beta2} we get
\begin{equation*}
\frac{1}{\eta}\ln \left(\frac{p_{\beta}}{p_{{T+\Gamma+1}}}\right) + \Delta' +\beta -(\alpha+1) \leq  
     c-1  + \Delta' + \frac{1}{\eta(c'-1)}\cdot \ln \left(\frac{{v_{max}\cdot e^{\eta\cdot (c-2)}}}{p_{min}}\right) \leq \Gamma
\end{equation*}
where the rightmost inequality follows from 
the assumption about
$\Gamma$.
\end{proof}
\end{proof}

\section{Lower Bounds}

In this section we show that certain limitations of our
main theorem are inevitable. 
We do so
by proving 
appropriate lower bounds.  
Our lower bounds apply 
more generally 
to classes of algorithms of which
the \EIP algorithm is a member, some
to all online algorithms (see \cref{sec:online-lb-size}) and others just for
the subclass of price-based algorithms (see \cref{sec:price-lb}).  Where there are gaps between our upper bound for the \EIP algorithm and our general lower bounds, we also show that the upper bound
is in fact the tight one for the \EIP algorithm.

Our lower bounds apply to algorithms that, like the \EIP algorithm, are allowed  both  some slackness $\Delta$ and some extension $\Gamma$ in order to
compete with the optimum schedule.  As these
two notions of augmentation may be traded-off against each other to some extent\footnote{{All our lower bounds 
hold even for a single horizon $T$ that is known to the algorithm in advance. In that setting, shifting between the two parameters 
is indeed possible.}
If there is a single horizon $T$ that is
known by the algorithm, the extension may be completely transferred to additional slackness by immediately scheduling at the last block all transactions that were intended for the extension period. {Similarly, slackness can be shifted to an extension when $T$ is known, by postponing some scheduling to the extension time.}}, we will generally prove lower bounds on their sum. 

\subsection{Augmentation of Price-Based Algorithms must depend on Value Range}\label{sec:price-lb}

In this section we prove 
that the dependence of $\Gamma$ or $\Delta$ on
the range of per-unit values is unavoidable for {\em price-based algorithms with good welfare guarantee}, proving \cref{prop:intro-value-range-lb}.

\begin{proposition}
    \label{prop:price-based-extension-lb} 
Fix any price-based algorithm with average block size limit $B$, slackness $\Delta$, and extension $\Gamma$. 
There exists an instance with values per unit in the range $[1,H]$ for which the algorithm does not obtain {better than} 
a  
$H^{1/4^{(\Gamma+\Delta)}}$  
fraction  of the optimal social welfare of a schedule with worst-case block size limit $B$ (with no slackness or augmentation).
Thus, if a price-based algorithm guarantees a constant fraction of this optimum, then $\Gamma+\Delta$ must grow as $\Omega(\log \log H)$.
\end{proposition}

\begin{proof}
    Given an algorithm with parameters $\Delta$ and $\Gamma$, 
    {and assume wlog that $B=1$,} we will
    choose some parameters $T, r$ and design a set of input scenarios 
    that for at least one of them the algorithm 
    cannot achieve more than 
    $(1+(\Gamma+\Delta)/T) / r$ fraction of the social
    welfare. This, with our choices of parameters $r$ and $T$, will give
    the desired upper bound.
    
    We will look at our algorithm running for $R=T+\Gamma$ steps, 
    where all transactions arrive at time $1$.   
    We will have $2^R+1$ possible inputs scenarios, 
    where in input scenario number $m$ ($0 \le m \le 2^R$)
    we have, for every integer $0 \le i \le m$, exactly $(R+\Delta)$ 
    transactions with value $r^i$ 
    (all arriving at time 1), each of size $1$.  
    We also specify how our adversary will respond
    to every price $p$ in each input scenario $m$.  If $p>r^m$ then there are no 
    transactions that are willing to pay this price so no transactions will be scheduled.
    if $p \le r^m$ then the adversary will always schedule only transactions with the lowest possible 
    value (i.e. those with value $r^i$ for the smallest
    $i$ so that $r^i\ge p$).  Note that there are enough transactions of this
    value to suffice for anything that can be scheduled by the end of time $R$.

    Now let us look at a price-based algorithm running for $R$ steps on one of these
    input scenarios.  For any price $p$ that is queried by the algorithm there are
    at most two possible answers by the adversary on any one of these scenarios: 
    either no transaction is scheduled or as many transactions as allowed by
    the block's capacity, each of the value $r^i$ for the $i$ completely determined
    by the price, is scheduled.  It follows that the complete set of possible
    behaviors of the algorithm on all possible scenarios can be described as a 
    binary tree
    of depth $R$, containing at most $2^R-1$ possible prices in all and splitting the
    possible transactions into the $2^R$ leaves of the tree.
    Since there are $2^R+1$ possible scenarios, two of them must map into the same
    leaf and thus produce the exact same schedule.  So let us assume that input 
    scenarios $m<m'$ get the same behavior by the algorithm, which must thus never 
    schedule any transactions with any value strictly 
    greater than $r^m$ as such do not exist in scenario $m$.  

    So now let us compare the social welfare of our algorithm on input scenario $m'$ to
    that of the optimal schedule. The optimal schedule for $T$ blocks
    with worst-case block size $1$
    always uses transactions with value
    $r^{m'}$ for a total welfare of $r^{m'} \cdot T$.  Since our algorithm cannot
    use transactions with value greater $r^m$ and has slackness $\Delta$
    and runs for $\Gamma$ extra steps, its welfare is bounded by 
    $r^m \cdot (T+\Gamma+\Delta)$.  The fraction of welfare obtained in scenario $m'$ 
    is bounded from above
    by $(1+(\Gamma+\Delta)/T) / r$.

    Now for our choice of parameters.  We let $T=\Gamma+\Delta$ which ensures that the numerator 
    in the expression above
    is the constant $2$, and thus, the fraction of social welfare (with respect to the benchmark of any schedule with worst-case block size limit $1$, and no slackness or augmentation) 
    that is achieved by the algorithm is bounded from above by $2/r$.  Our choice of $T$ also ensures that $R = T + \Gamma\le 2T = 2(\Gamma+\Delta)$
    and since for our choice of $r$ we only needed that $r^{2^R} \le H$ we may choose 
    $r=H^ {1/2^{2T}}=H^{1/2^{2(\Gamma+\Delta)}}= H^{1/4^{(\Gamma+\Delta)}}$, 
    which concludes the proof.
\end{proof}

\subsection{Online Algorithms need to Relax the Maximum Block Size}\label{sec:online-lb-size}

In this section we restate and prove \cref{prop:intro-lb-block}, showing that every online algorithm with maximum block size $c\cdot B$ for $c<2$, must lose a constant fraction of welfare.

\begin{proposition} \label{c2}
   Any online algorithm that produces schedules with maximum block size $c \cdot B$, for some {$c$ with $1\leq c<2$,} 
   and average 
   block size limit $B$ with slackness $\Delta(T) = o(T)$, even with an extension $\Gamma(T)=o(T)$, 
   must lose at least 
   $\delta_0=\Omega(2-c)$ 
   fraction of social welfare relative to schedules with worst-case block size limit $B$ 
   (even when the values of all transactions are in \{1,2\}). 
\end{proposition}

{\bf Remark:} this proposition applies to the case where the largest transaction size can be as large as $B$.  
If we have an upper bound on $q_{max}/B$ then
the lower bound applies whenever $c < 1 + q_{max}/B$, matching the positive result in \cref{sec:greedy}.

\begin{proof} 
    Here is the construction for the lower bound (assuming w.l.o.g. that $B=1$): 
    We take $T$ to be large enough so that $\Gamma$ and $\Delta$
    that are $o(T)$ are as small as we want relative to $T$.
    In each of 
    the first $T/2$ steps we get one 
    ``red'' transaction with size $1$ and value 1 and one
    ``green'' transaction with size $c/2+\varepsilon$ and value 
    $2$.  
    Notice that as $c>1$ at most one of these transactions can fit into each block.  
    Let us look at $G$ which is the number of green transactions that were scheduled
    during the first $T/2$ steps.

    Case I: $G \le T/4$.  In this case the total value achieved in the 
    first $T/2$ steps is at most 
    {$2 \cdot G+1 \cdot (T/2-G)\leq T/2+G \leq 3 T/4$.} 
    In this case consider the case where for the next $T/2$ steps, 
    in each step a new transaction of size 1 and value $2$ enters.
    In this case the best that the mechanism can do is get value $2$ for each
    of the remaining $T/2$ steps and the extension period $\Gamma = o(T)$ for a total social welfare of at most 
    $7T/4 + o(T)$.  
    In comparison the optimal schedule with block size 1 would schedule all the green transactions in the first half, getting a value of $2$ every step for a total of $2T$,
    so the ratio of the two welfare expressions is bounded
    from above by $7/8 +o(1)$.

    Case II: $G>T/4$.  In this case consider the case where for the next $T/2$ steps we get 
    a large number of small transactions of size $\gamma$ with 
    value $\gamma$ (i.e. value per unit of $1$) each.
    {An optimal schedule will schedule red transactions in the first $T/2$ rounds for value of $T/2$, and all $T/2$ green transactions in the second $T/2$ rounds for value of $2\cdot T/2=T$. Additionally, all leftover capacity in the blocks containing green transactions, which have total capacity of $T/2\cdot (1-c/2-\varepsilon)$, will be filled with the small transactions, getting value of 1 per unit. So the total welfare will be $3T/2+T/2\cdot (1-c/2-\varepsilon) = T\cdot (2-c/4-\varepsilon/2) $.
    The algorithm is getting value of only $2\cdot G+1\cdot(T/2-G)$ in the first $T/2$ rounds, and value of only $(2+(1-(c/2+\varepsilon))\cdot (T/2-G)+1\cdot G$ in the second $T/2$ rounds. 
    It can only get an additional $o(T)$ due to the slackness and extension.  
    So the total value is at most 
    $3T/2 + (1-c/2-\varepsilon)\cdot (T/2-G) +o(T) < 3T/2 + (1-c/2-\varepsilon)\cdot T/4 +o(T) = T\cdot (7/4-c/8+O(\varepsilon)+o(1))$.
    We conclude that the ratio of the welfare of the algorithm to the welfare of the optimum is at most 
    $\frac{T\cdot (7/4-c/8+O(\varepsilon)+o(1))}{T\cdot (2-c/4-\varepsilon/2)} = \frac{(14-c+O(\varepsilon)+o(1))}{(16-2c-4\varepsilon)} =
     \frac{((16-2c)-(2-c)+O(\varepsilon)+o(1))}{(16-2c-4\varepsilon)}\leq  1- \frac{2-c}{14} + O(\varepsilon)+o(1)
    $ (as $2>c\geq 1$). 
    And thus the loss is at least $\frac{2-c}{14}- O(\varepsilon)-o(1)$.
    }

\end{proof}

\subsection{Lower Bounds Examples for the \EIP Algorithm}\label{sec:EIP-LB}

Our lower bounds for general classes of algorithms presented above do not match the upper bounds of our main theorem 
for the \EIP algorithm (\cref{thm--ub--eip1559-optimal-frac}).  
In this section we observe, by providing examples, that the upper bounds of our main theorem are in fact tight for  the \EIP algorithm. 

\subsubsection{$c=2$ Does not Suffice for the \EIP Algorithm}

In \cref{c2} we showed that the maximum block size must be chosen to be at least
$2\cdot B$ in any online algorithm that attempts to compete with offline
schedules whose worst case block size is $B$ (when $q_{max}=B$), i.e., in our notation,
the lower bound showed that we need $c \ge 2$. 
\cref{warmup} showed that 
$c=2$ indeed suffices for the online algorithm that we presented.  However,
our main theorem for $q_{max}=B$ 
requires the maximum block size to be strictly greater than twice the
average block size {limit}, $c > 2$, and in fact the result also 
uses an  extension that grows in rate that is 
inversely proportional to $c-2$.
So the question is whether having the maximum block size $2 \cdot B$ in the \EIP algorithm, i.e. $c=2$, suffices
for near-optimality.  This question is especially interesting since the actual choice of parameters 
in Ethereum is exactly $c=2$, as the target block size is {$B=18M$} 
gas and the maximum block size
is defined to be {$2B=36M$} 
gas.

Here is an example showing that
the \EIP algorithm with $c=2$, even with an extension that grows as $\Gamma(T) =o(T)$,
 cannot give even half of the
optimum social welfare.\footnote{
This example also 
shows that even when $c>2$, as $c$ approaches $2$, the extension length must grow inversely proportional to $c-2$.  Simple tweaks in the example will show that even longer extensions cannot guarantee even a lower fraction of the welfare.}
}
For notional convenience let us assume without loss of generality that $p_{min}=1$ 
(otherwise just scale all the values in the rest of the proof by $p_{min}$).
Our example starts with steps where no transactions arrive, and so blocks are empty and the price
decreases to $p_{min}=1$.  
From this point on, every time step $t$,
we get as input one ``high'' transaction
of size exactly $B$ and value per unit $10$, and additional ``low'' transactions of total size $(1+\varepsilon)\cdot B$
and value per unit $2$, where $\varepsilon>0$ is a parameter, to be specified below, that may be as small as we wish.  
(The size composition of the low transactions does not matter as long as their total
size is $(1+\varepsilon)\cdot B$.)  In our example the low transactions offer a high tip, i.e. the adversary
always schedules them before the high transaction. 

Now let us see how the \EIP algorithm deals with this input.  As long as $p_t \le 2$, both the high and the low
transactions can be scheduled, and since the 
low transactions have priority they will be scheduled first, filling the block to size $(1+\varepsilon)\cdot B$.
At this point the high transaction cannot get scheduled, since scheduling it will just exceed the maximum
allowed block size {of $2\cdot B$}, so the block closes at size $(1+\varepsilon)\cdot B$.  Thus the next price will be
$p_{t+1} = p^t \cdot e^{\varepsilon \eta}$.  Notice that at this rate of price increase, it will take 
$\Omega(1/(\varepsilon\eta))$ blocks for the price 
to reach 
the value 2. 
As long as this is the case,
the \EIP algorithm's welfare per block is less than a quarter of 
the optimal schedule's welfare per block (since the optimal schedule can at the very least accept a high 
transaction of value 10 every block
while the algorithm accepts transactions of value $2 \cdot (1+\varepsilon)$ which for small enough $\varepsilon$
is less than $2.5=10/4$).
Now, if we take $T$ to be half of the time it took for the price to rise to $2$ (so, $T=\Omega(1/(\varepsilon\eta))$)
we must have the extension also be at least $\Gamma \ge T$ in order to reach even half of the optimal social
welfare (since the algorithm that accrues welfare at a rate of a quarter of the optimal, must run at least for 
twice as many blocks in order to reach half the social welfare).  In order to show that no
extension $\Gamma(T)=o(T)$ suffices to get even half of the social welfare on this example,
just choose $\varepsilon = O(1/(\eta T))$
for $T$ so that $\Gamma(T) \le T$.\footnote{When $c>2$ but is close to 2, one may take this example
with $\varepsilon \approx c-2$
and get a lower bound of $\Omega\left(\frac{1}{\eta(c-2)}\right)$ on the extension.}

Thus, this lower bound shows that Ethereum's choice of parameter $c=2$ is problematic when 
there are transactions of size $q_{max}=B$ that contribute significantly to the welfare. 
As the general version of our main theorem showed, this problem disappears once transaction sizes are limited by some $q_{max}<B$.

\subsubsection{\EIP suffers a single logarithmic 
loss in the range size}

While, as we showed in \cref{warmup}, online algorithms may compete with the optimum
using constant slackness and extension ($\Delta=1$ and $\Gamma=1$), 
\cref{prop:price-based-extension-lb} proved a lower bound of 
$\Delta+\Gamma = \Omega(\log\log H)$ for price-based algorithms when the transaction per-unit values
are in the range $[1,H]$.  Our upper bound for the \EIP algorithm 
in the main theorem is exponentially worse, with
{$\Gamma=O(\log H)$. Additionally, by \cref{prop:EIP-slack} we know that  the \EIP algorithm has {slackness} $\Delta=O(\log H)$.

We show that these upper bounds
are the correct ones for the \EIP algorithm.
Our example starts with steps where no transactions arrive, and so blocks are empty and the price
decreases to $p_{min} \le 1$.  
{At this point in time,  a large amount of 
``high'' transactions with value per unit of $v=H$ arrive. Additionally, at that time also arrive a large amount of 
``low'' transactions with value of $L$ per unit,} 
where $H>L>1$ are parameters.
The exact size composition of low and high transactions does not matter, as
long as there is enough demand in each of these values to always completely
fill
blocks to the maximum block size.  For our example, the low transactions
will offer high ``tips'', i.e. the adversary will give them priority.

So how will the \EIP algorithm react to this input?  After the price decreases to 
$p_{min}\leq 1$, only low transactions are scheduled, filling blocks to the
maximum block size, so the prices will increase as 
$p_{t+1}=p_t \cdot e^{(c-1)\eta}$, where $c$ is the maximum block size 
multiplier.  So after $k$ steps we have 
$p_{t+k} = p_t \cdot e^{k\cdot (c-1)\eta}$.  As long as $L \ge p_t$ the adversary will only schedule
low transactions, while the optimum schedule will only have high
transactions, getting a ratio of $H/L$ in the social welfare. This
can go on for at most $k=\Theta(\log L/((c-1)\eta))$ steps (after which $p_t$ that started at $p_{min} \le 1$ can
exceed $L$), which
by choosing $L$ to be an arbitrary constant factor 
lower than $H$ (so
$\log H = \Theta(\log L)$)
gives us that extension 
$\Gamma \ge k = \Omega(\log H/((c-1)\eta))$ is needed before
we get better than the constant 
factor approximation $H/L$ to the welfare.  Also,
looking at these $k$ steps of price increase, note that all blocks
are completely full, i.e. their total size is $k\cdot c\cdot B$, which
implies that the slackness is 
$\Delta \ge k \cdot (c-1) = \Omega(\log H/\eta)$.

\section{Model Extensions}\label{sec:ext}
In this section we consider two model extensions: to partially patient bidders and to multiple resources. In both cases we present impossibilities, showing that losing a constant fraction of the welfare is inevitable even for general online algorithms (not only price-based algorithms or the \EIP algorithm).   

\subsection{Partially Patient Bidders: Online Algorithms must Lose a Constant Factor}

All our results so far assumed ``patient bidders'', i.e. where a transaction's value
for its user remains the same over time.  The opposite assumption of ``impatient
bidders'' assumes that a transaction {\em must} be scheduled in the immediate block
or it loses all value for its user.  A more general model would capture some sensitivity of the value of a transaction on the time of its execution, where
the value decreases with time.  There are two simple models for such dependence on time:

\begin{definition}
    A transaction with value $v$ and {\em discount factor} $\rho$ 
    (where $0 \le \rho < 1$) 
    has value for its
    user of 
    $v \cdot (1-\rho)^{t_e-t_a}$, where $t_a$ is its arrival time and $t_e$ its execution time.
\end{definition}

\begin{definition}
    A transaction with value $v$ and {\em patience} $p$ has value of $v$ for its
    user if it has been scheduled within $p$ blocks of its arrival time and value $0$ otherwise.  I.e. if $t_a$ is its arrival time and $t_e$ its execution time then
    the user's value is $v$ if $t_a \le t_e \le t_a+p$ (and $0$ otherwise).
\end{definition}

Thus the fully patient model corresponds to patience level $p\rightarrow\infty$ or 
discount factor $\rho=0$, while the 
impatient model corresponds to $p=0$ or $\rho\rightarrow 1$.
Clearly the EIP-1559 algorithm does not allow its input to specify any patience level 
or discount factor.  It is not hard to observe that if we just run EIP-1559 
on partially patient bidders without taking the patience or discount factor 
into account then since it does not distinguish between ``new'' transactions and
old ones that already lost most or all of their value, it cannot produce efficient results.  But perhaps if we just let the algorithm at every block take into account the
{\em current} value of the transaction, then we regain efficiency?  I.e. for the
discount factor model, at block $t$,  
a pending transaction $i$ that has already arrived is considered for allocation if and only if 
$v_i \cdot (1-\rho_i)^{t-t_i} \ge p_t$.  Unfortunately, we get a negative answer and,
in fact, lower bounds for any online algorithm.  Significantly, these lower bounds
hold for arbitrarily high patience (or low discount factor).

\begin{proposition}
    Fix an online algorithm with average block size limit $B$, slackness $\Delta$,
    and extension $\Gamma$, where 
    $\Delta(T)+\Gamma(T)=o(T)$. 
    For every {maximum} discount rate ${\rho_{max}}>0$ 
    there exist a
    time horizon $T$, and an input sequence where
    {every bidder $i$ has a 
    non-negative discount rate $\rho_i \in \{0,\rho_{max}\}$ 
    and value $v_i \in \{1,2\}$,} for which
    the algorithm 
    loses 
    at least a fraction $\delta_0=1/20-o(1)$ of 
    welfare relative to the optimal schedule with
    worst-case block size limit $B$.   
\end{proposition}

\begin{proof}
    Fix an online algorithm with average block size {limit} $1$, slackness $\Delta(T)$,
    and extension $\Gamma(T)$.
    Choose $T$ that is larger than $\Gamma(T)+\Delta(T)$ by a large factor $K$,
    and also such that $(1-{\rho_{max}})^{T/3} \le 1/2$.
    Let $p=T/3$
    and denote $\rho = {\rho_{max}}$.
    Now consider the following input scenario: At block 1, 
    we have $p$ ``patient'' transactions
    arriving, each with size $q_i=1$, value $v_i=2$, and discount rate $\rho_i=0$.
    Also at each of the first $p$ blocks arrives one ``hasty'' transaction, each with 
    size $q_i=1$, value $v_i=1$, and discount rate $\rho_i=\rho$.  So now our 
    algorithm has to decide what to schedule during the first $p$ steps: the high-value
    patient transactions or the low-value hasty transactions. 

    Case I: by the end of time $p$ at least $p/2$ hasty transactions have been
    scheduled.  In this case, for block $p+1$ our adversary will choose 
    as input $2p$ more transactions each with size $q_i=1$, value $v_i=2$, and 
    discount rate $\rho_i=0$.  Now by the end of time $T=3p$ at most
    $T+\Delta(T)+\Gamma(T) = T \cdot (1+O(1/K))$ transactions have been scheduled, 
    at least $p/2$
    of them of value $1$ and the rest of value $2$.  
    So the total social welfare achieved is at most 
    $2 \cdot (5/2+O(1/K)) p + 1 \cdot p/2 = (11/2+O(1/K))p$.
    On the other hand, the optimal offline benchmark would only schedule the 
    high-value patient transactions for a social welfare of $2 \cdot 3p = 6p$, so
    the ratio of the two welfare terms is $11/12+O(1/K)$ which, by choosing $K$ large enough, 
    is less than $19/20$.

    Case II: by the end of time $p$ at most $p/2$ hasty transactions have been
    scheduled.  In this case, in each of the next $p$ steps we get a single 
    transaction
    with size $q_i=1$, value {$v_i=2$}, and 
    discount rate $\rho_i=\rho$ (and nothing more arrives in the last $p$ steps).
    We can assume wlog that our online algorithm was ``clever enough'' to schedule
    all new transactions during the ``middle'' $p$ steps as
    otherwise switching a scheduled hasty transaction or patient transaction
    with the new transaction that
    arrived at that time step will only increase the social welfare (since both the hasty
    transaction and the new transaction lose value at the same rate $\rho$ but the
    new transaction {has higher current value}. Patient transactions do not lose value at all.)  
    It follows that at most
    $\Delta$ hasty transactions were scheduled 
    during the ``middle'' $p$ blocks and so the social 
    welfare that our algorithm gets can be bounded 
    from above by $(19/4 + o(1)) \cdot p$:
    $2p$ from all the new transactions (that are scheduled in the middle $p$ steps) 
    plus $2p$ from all the patient transactions plus at most $3p/4+\Delta$ from the hasty transactions since at least
    half of them (minus $\Delta$) 
    are scheduled in the last $p$ blocks after they have lost at least half their value (since each of them waited at least for the middle $p$ steps and
    $(1-\rho)^p \le 1/2$).  The optimal offline solution is to schedule the hasty
    transactions as they come in the first $p$ steps, schedule the 
    new transactions, as
    they come, in the middle $p$ steps, and leave the patient transactions for the
    last $p$ steps, getting the full value of $5p$.  
    The ratio between these
    welfares is at most $19/20 + o(1)$, {and thus the loss is at least $1/20-o(1)$}. 
\end{proof}

Notice that the proof uses transactions with different discount rates.  
It is
interesting to ask what can be done when all transactions share the same, global, discount
rate.  We conjecture that, in fact, the modified EIP-1559 (that takes into account 
current values) does produce nearly-optimal schedules with a loss that approaches 
the loss of the fully patient case, as the {maximum} 
discount rate ${\rho_{max}}$ approaches 
$0$.

For the other model, of bidders with given patience levels, we
can show that having a global patience level does not suffice and 
exhibit an
impossibility result for online algorithms where all transactions share 
a fixed, arbitrarily high, patience level $p_i=p$.

\begin{proposition}
    Fix an online algorithm with average block size limit $B$, slackness $\Delta(T)=o(T)$,
    and extension $\Gamma(T)=o(T)$
    then {there exists $p^*$ such that
    for every global patience level $p \ge p^*$ and every sufficiently large time horizon
    $T \ge 2(p+1)$},
    there exists an input sequence where
    every bidder $i$ has 
    patience $p_i=p$ and value $v_i\in \{1,2\}$,
    for which the algorithm loses at least a fraction 
    $\delta_0=1/10$ of welfare relative to the optimal schedule with
    worst-case block size limit $B$.
\end{proposition}

{\bf Comment:} If the patience of all bidders is just slightly larger than the 
time horizon, $p_i \ge T$,
then we are back to the fully patient case and near-optimality is regained.

\begin{proof}
    {Without loss of generality assume that} $B=1$ and fix an online algorithm with average block size limit $B$ and slackness $\Delta(T)$ and extension $\Gamma(T)$. {Let $p^*$ be such that
    $\Delta(T)+\Gamma(T) \le \varepsilon \cdot T$ for all $T \ge p^*$ for
    some small $\varepsilon$ to be chosen below.
    Given $p \ge p*$,} 
    we will prove the claim for $T=2(p+1)$ and the same applies for any larger $T$ by adding an arbitrarily
    long time prefix with no input transactions.
    
    At time 1 arrive $p+1$ ``green'' transactions with value 1 and size 1. Additionally, every time step for the next $p$ steps
    arrives a single ``red'' transaction with value 2 and size 1.  
    Let us look at how many red transactions
    were scheduled by our algorithm by the end of time $p+1$.

    Case I: at least $(p+1)/2$ red transactions were scheduled.  Since the total number of
    scheduled transactions by this time is at most $(p+1)+\Delta(T)$, 
    this means that at least $(p+1)/2-\Delta(T)$ green
    transactions were not scheduled and can never be scheduled any more since their deadline has passed.
    Now assume that no more transactions arrive in the next $p+1$ steps, so our algorithm can schedule
    the rest of the red transactions but, again, not the unscheduled green ones.  This gives 
    an upper bound of $2(p+1)+(p+1)/2+\Delta(T)=5(p+1)/2+\Delta(T)$ on the social welfare of our online algorithm.  
    The optimal offline
    algorithm could have scheduled only the green transactions during the first $p+1$ steps and then
    the red transactions, one after another, after that for a total social welfare of ${2p + (p+1)= 3p+1}$.
    The ratio between these amounts is $5/6+O(\varepsilon)$, {and thus the loss is at least $1/6-o(1)>1/10$}.

    Case II: less than $(p+1)/2$ red transactions were scheduled by the end of time $p+1$.  Now assume
    that another $p+1$ ``new'' transactions arrive at time ${p+2}$ each with value 2 and size 1.  
    Since the total quantity scheduled in the first $p+1$ steps is at most $(p+1)+\Delta(T)$ and
    since at most $(p+1)/2-1$ of this quantity has a value per unit of 2 and the rest have 
    value 1, then during the first
    $(p+1)$ steps our online algorithm could have made at most $(p+1) + \Delta(T) + (p+1)/2-1$ social welfare.
    In the second half, including the extra allowed $\Gamma(T)$ steps, an upper bound to the
    social welfare that it can make is 
    $2((p+1)+\Delta(T)+\Gamma(T))$.  Thus the total
    welfare of our online algorithm is bounded from above by
    $7(p+1)/2 + O(\varepsilon\cdot p)$ welfare.
    The optimal offline
    algorithm could have scheduled a green transaction in round 1, then the $p$ red transactions during the next $p$ steps and then 
    scheduled the new transactions in the next $(p+1)$ steps for a total welfare of $1+2p+2(p+1) = 4p+3$.  
    The ratio between these amounts is $7/8+O(\varepsilon)$.  Choosing $\varepsilon$ small
    enough that this amount is less than $9/10$ (and thus the loss is at least $1/10$) 
    concludes the proof of the theorem.
\end{proof}

\subsection{Multidimensional Case: Online Algorithms Must Lose a Constant Factor}

In our basic model each transaction has a ``single-dimensional'' size $q_i$.  More
generally one may consider a model where there are $m$ different resources, each
block has a size limit $B_j$ for every resource $j$ and each transaction $i$ uses
the amount $q_{ij}$ of every resource $j$ (see e.g. \cite{MDBFM23,ADM24}).  
The natural generalization of the slackness condition to the multi-dimensional case
would require that every $T$ consecutive blocks use at most $(T+\Delta) \cdot B_j$
amount of each resource $j$.  The introduction of the ``blob'' resource to Ethereum
\cite{4844} is a step in this direction.
While one may hope to extend the near-optimality result to such a multi-dimensional model,
it turns out any online algorithm must lose some constant fraction of social
welfare in the multi-dimensional case.

\begin{theorem}
If there are at least three resources then any online algorithm with average block
size $B_j$ and slackness $\Delta(T)=o(T)$ 
must lose at least a constant fraction 
$\delta_0=1/6-o(1)$
of social welfare on some input sequence,
even when allowed $\Gamma(T)=o(T)$ extension.
\end{theorem}

\begin{proof}
    Consider the case of three resources, X, Y, and Z, and assume that the
    block size for each resource is $B_X=B_Y=B_Z=1$.
    Let $t=T/2$ where $T$ is the horizon over which our online algorithm
    aims to approximate OPT and fix an online algorithm that uses each 
    resource at total capacity of at most $(t+\Delta)$ during 
    any consecutive $t$ blocks.  
    
    At time 1 we get $t$ demand for
    the bundle $\{X,Z\}$ and $t$ 
    demand for the bundle $\{Y,Z\}$ each at value of 1 per unit of the bundle.  
   
    During the first $t$ steps
    there is only capacity of $t+\Delta$ of resource $Z$ so at least 
    one of the two types of
    demanded bundles is allocated at most at total capacity $(t+\Delta)/2$.  
    Without loss of generality assume that the bundle $\{X,Z\}$ is allocated at most 
    $(t+\Delta)/2$ during the first $t$ steps.  Now at time $t+1$ we get
    another demand of $t$ for the singleton bundle $\{X\}$ again at value 1 per unit of 
    $X$.  (Of course,
    had the bundle $\{Y,Z\}$ been the one that was at least half-un-allocated then
    our adversary would have brought demand for the singleton bundle $\{Y\}$ instead.)
    The best that our online
    algorithm can do at this point is pack singleton bundles $X$ together 
    with any leftover of the original
    bundles $\{Y,Z\}$ into the remaining blocks, and leaving at least
    $t/2-\Delta-\Gamma$ demand for the bundle $\{X,Z\}$ unsatisfied for a total value
    of at most $(5t+\Delta+\Gamma)/2$.
    OPT, on the other hand, could have fully
    allocated everything by allocating only the $\{X,Z\}$ bundles for the first $t$
    blocks and packing the demand for $\{Y,Z\}$ with $\{X\}$ in each 
    of the last $t$ blocks getting
    a total value of $3t$.  The ratio between these two levels of social welfare is
    $5/6+o(1)$.
\end{proof}

\DOUBLEBLIND{
\section*{Acknowledgments} 
We thank Shiri Ron for useful discussions during the early stages of this work.  Noam Nisan thanks his colleagues in Starkware for many useful discussions.  Moshe Babaioff's research is supported by a Golda Meir Fellowship and the Israel Science Foundation (grant No. 301/24). Noam Nisan's research was supported by a grant from the Israeli
Science Foundation (ISF number 505/23).
}

\bibliographystyle{alpha}
\bibliography{bib}

\end{document}